\documentclass[12pt,twoside]{article}

\usepackage{mathtools}
\usepackage{lipsum}

\usepackage{multirow}
\usepackage{makecell}
\usepackage{graphicx}%
\usepackage{multirow}%
\usepackage{amsmath,amssymb,amsfonts}%
\usepackage{amsthm}%
\usepackage{mathrsfs}%
\usepackage[title]{appendix}%
\usepackage{xcolor}%
\usepackage{textcomp}%
\usepackage{manyfoot}%
\usepackage{booktabs}%
\usepackage{algorithm}%
\usepackage{algorithmicx}%
\usepackage{algpseudocode}%
\usepackage{listings}%
\usepackage{manyfoot}
\usepackage{xcolor}
\usepackage{lscape}
\usepackage{graphicx}%
\usepackage{multirow}%
\usepackage{multicol}%
\usepackage{amsmath,amssymb,amsfonts}%
\usepackage{amsthm}%
\usepackage{mathrsfs}%
\usepackage[title]{appendix}%
\usepackage{xcolor}%
\usepackage{textcomp}%
\usepackage{manyfoot}%
\usepackage{booktabs}%
\usepackage{algorithm}%
\usepackage{algorithmicx}%
\usepackage{algpseudocode}%
\usepackage{listings}%
\usepackage{array}
\usepackage{enumerate}
\usepackage{hyperref}
\usepackage{amsmath, nccmath}
\usepackage{ mathrsfs }
\usepackage{autobreak}
\usepackage{mathtools}
\usepackage{upgreek}
\usepackage{multirow}
\usepackage{nicematrix}
\usepackage{tabularray}
\usepackage{adjustbox}
\usepackage{ tipa }
\usepackage{tabularx}
\usepackage{tabulary}
\usepackage{ragged2e}
\usepackage{array}
\usepackage{multicol}
\usepackage{nicematrix}
\usepackage{booktabs}
\usepackage{xcolor}
\usepackage{upgreek}
\usepackage{autobreak}
\usepackage{makecell}
\usepackage{float}
\usepackage{mathrsfs}
\usepackage{amsthm}
\usepackage{enumerate}
\usepackage{hyperref}
\usepackage[a4paper]{geometry}
\usepackage{amsmath, nccmath}
\usepackage{ mathrsfs }
\usepackage{autobreak}
\usepackage{mathtools}
\usepackage{upgreek}

\makeatletter
\renewcommand\title[1]{\gdef\@title{\reset@font\Large\bfseries #1}}
\renewcommand\section{\@startsection {section}{1}{\z@}%
                                   {-3.5ex \@plus -1ex \@minus -.2ex}%
                                   {2.3ex \@plus.2ex}%
                                   {\normalfont\large\bfseries}}
\renewcommand\subsection{\@startsection{subsection}{2}{\z@}%
                                     {-3ex\@plus -1ex \@minus -.2ex}%
                                     {1.5ex \@plus .2ex}%
                                     {\normalfont\normalsize\bfseries}}
\renewcommand\subsubsection{\@startsection{subsubsection}{3}{\z@}%
                                     {-2.5ex\@plus -1ex \@minus -.2ex}%
                                     {1.5ex \@plus .2ex}%
                                     {\normalfont\normalsize\bfseries}}

\def\@runningauthor{}\newcommand{\runningauthor}[1]{\def\runningauthor{#1}}
\def\@runningtitle{}\newcommand{\runningtitle}[1]{\def\runningtitle{#1}}

\g@addto@macro\bfseries{\boldmath}

\makeatother
\usepackage{mathtools}

\usepackage{amsthm,amsmath,amssymb}
\usepackage{graphicx}

\theoremstyle{plain}
\newtheorem{theorem}{Theorem}[]
\newtheorem{lemma}{Lemma}[]

\theoremstyle{definition}

\newtheorem{example}{Example}[]

\theoremstyle{remark}
\newtheorem{remark}[]{Remark}

\title{Binary cyclic codes from permutation polynomials over $\mathbb{F}_{2^{m}}$}

\runningtitle{Binary cyclic codes from permutation polynomials over $\mathbb{F}_{2^{m}}$}

\author{Mrinal Kanti Bose\thanks{First Author is supported by funding organisation Indian Institute of Technology (ISM) Dhanbad.}\\
\small Department of Mathematics and Computing\\[-0.8ex]
\small Institute of Technology (ISM) Dhanbad,\\[-0.8ex] 
\small Dhanbad, Jharkhand, India\\
\small\tt 21dr0111@mc.iitism.ac.in
\and
Udaya Parampalli 
\\
\small School of Computing and Information Systems\\[-0.8ex]
\small The University of Melbourne,\\[-0.8ex]
\small Parkville, Victoria, Australia\\
\small\tt udaya@unimelb.edu.au
\and
Abhay Kumar Singh 
\\
\small Department of Mathematics and Computing\\[-0.8ex]
\small Institute of Technology (ISM) Dhanbad,\\[-0.8ex]
\small Dhanbad, Jharkhand, India\\
\small\tt abhay@iitism.ac.in
}

\runningauthor{Mrinal\ Kanti\ Bose, Udaya \ Parampalli, Abhay\ Kumar\ Singh}

\date{}

\begin{document}
\allowdisplaybreaks

\maketitle

\thispagestyle{empty}

\begin{abstract}
Binary cyclic codes having large dimensions and minimum distances close to the square-root bound are highly valuable in applications where high-rate transmission and robust error correction are both essential. They provide an optimal trade-off between these two factors, making them suitable for demanding communication and storage systems, post-quantum cryptography, radar and sonar systems, wireless sensor networks, and space communications. This paper aims to investigate cyclic codes by an efficient approach introduced by Ding \cite{SETA5} from several known classes of permutation monomials and trinomials over $\mathbb{F}_{2^m}$. We present several infinite families of binary cyclic codes of length $2^m-1$ with dimensions larger than $(2^m-1)/2$. By applying the Hartmann-Tzeng bound, some of the lower bounds on the minimum distances of these cyclic codes are relatively close to the square root bound. Moreover, we obtain a new infinite family of optimal binary cyclic codes with parameters $[2^m-1,2^m-2-3m,8]$, where $m\geq 5$ is odd, according to the sphere-packing bound.
\end{abstract}

\textbf{Keywords:} Cyclic code; Permutation polynomial; Linear span; Sequence
\vskip 3pt
\textbf{MSC:} 94B15, 11T71, 11T06

\section{Introduction}\label{sec1}
Let $p$ be a prime and $q=p^m$, where $m$ is a positive integer. Let $\mathbb{F}_{q}$ be a field with $q$ elements. We call a polynomial $f(x)\in \mathbb{F}_{q}[x]$ a \textit{permutation polynomial} (PP) of $\mathbb{F}_{q}$ when the evaluation map $f:a\mapsto f(a)$ is a bijection. A linear $[v, k, d]$ code $\mathcal{C}$ of length $v$ is a $k$-dimensional subspace of $\mathbb{F}_{p}^{v}$ equipped with a minimum nonzero Hamming distance $d, ~d \ge 3$. A linear $[v,k,d]$ code over $\mathbb{F}_{p}$ is said to be optimal if there is no such $[v,k,d^{'}]$ code with $d^{'}\geq d+1$. 
A linear code $\mathcal{C}$ over $\mathbb{F}_{p}$ is said to be cyclic if a codeword $(a_0,a_1,\dots,a_{v-1})\in\mathcal{C}$ implies that its cyclic shift $(a_{v-1},a_0,\dots,a_{v-2})\in\mathcal{C}$. 
\textcolor{black}{Cyclic codes have a simple representation in terms of ideals in the polynomial algebra $\mathbb{F}_{p}[x].$
	Then,} we can identify any codeword $(a_0,a_1,\ldots,a_{v-1})\in \mathcal{C}$ with the polynomial $\sum_{i=0}^{v-1}a_{i}x^i \in \mathbb{F}_{p}[x]/(x^v-1)$. As we know, if $\operatorname{gcd}(v,p)=1$, then $\mathbb{F}_{p}[x]/(x^v-1)$ is a principal ideal ring and the cyclic code $\mathcal{C}$ of length $v$ is an ideal of the ring $\mathbb{F}_{p}[x]/(x^v-1)$. We use the notation $\langle g(x)\rangle$ to denote a principal ideal of $\mathbb{F}_{p}[x]/(x^v-1)$, where $g(x)$ is a monic polynomial of least degree in that principal ideal. Let $\mathcal{C}=\langle g(x)\rangle$ be a cyclic code, where $g(x)$ is called the generator polynomial of $\mathcal{C}$.  The dual code of $\mathcal{C}$ is also cyclic, denoted by $\mathcal{C}^{\perp}$. Let $h(x)=(x^{v}-1)/g(x)$ is called the check polynomial of $\mathcal{C}$ and let $h^{*}(x)$ be the reciprocal of $h(x)$. The dual code $\mathcal{C}^{\perp}=\langle h^{*}(x)\rangle$. 
\vskip 1pt
Another useful representation for cyclic codes is through the trace functions and an infinite sequence. In Section 3 of \cite{SETA5}, Ding defined a sequence $s^{\infty}=(s_{t})_{t=0}^{\infty}$ of period $v$ over $\mathbb{F}_{p}$ from an arbitrary polynomial $F(x)$ over $\mathbb{F}_{p^m}$ as
\begin{equation}\label{B}
	s_{t}=\operatorname{Tr}\left(F(\alpha^{t}+1)\right)\textnormal{ for all $t\geq 0$,}
\end{equation}
where, $\alpha$ is a primitive element of $\mathbb{F}_{p^m}$ and $\operatorname{Tr}(x)=\sum_{i=0}^{m-1}x^{p^{i}}$ is the trace map from $\mathbb{F}_{p^m}$ to $\mathbb{F}_{p}$. The cyclic code generated by the minimal polynomial of the sequence $s^{\infty}$ is denoted by $\mathcal{C}_{s}$. 
Ding \cite{SETA5}, and Ding and Zhou \cite{SETA3} raised questions on how to choose the polynomial $F(x)$ over $\mathbb{F}_{p^m}$ that could give optimal parameters on the cyclic codes and produced many optimal and almost optimal cyclic codes by employing several known families of almost perfect nonlinear (APN) and perfect nonlinear (PN) monomials and trinomials over binary and nonbinary fields. Subsequently, Tang et al. \cite{SETA11} solved two open problems on cyclic codes presented in \cite{SETA3} and \cite{SETA5}. Notably, Rajabi and Khashyarmanesh \cite{SETA10} extended earlier results on the construction of cyclic codes and solved two open problems proposed in \cite{SETA5}; Li et al. \cite{SETA9} provided partial answers for an open problem proposed in \cite{SETA3}. Mesnager et al. \cite{SETA8} complemented some earlier results and studied cyclic codes from several known families of low differential uniform monomial functions and provided partial answers to three open problems proposed in \cite{SETA3,SETA5}. Recently, Xie et al. \cite{SETA22} employed two classes of sequences to construct four families of binary cyclic codes and showed the existence of some codes with minimum distance satisfies the square root bound. A recent survey on the impressive developments in the last decade in the direction of a sequence construction of cyclic codes over finite fields can be found in \cite{SETA12, SETA14}.
\vskip 1pt
\textcolor{black}{The development given above shows the popularity of the trace sequence approach to determine interesting cylic codes $\mathcal{C}_{s}$ over $\mathbb{F}_{p}$. However, there is not much attention to binary cyclic codes $\mathcal{C}_{s}$ by choosing suitable permutation trinomials over $\mathbb{F}_{2^m}$ in the literature. 
	It is a challenging question of how to choose specific trinomials 
	to design an infinite family of optimal binary cyclic codes meeting certain bounds or an infinite family of binary cyclic codes with dimensions larger than half of the length of each code in the family with a minimum distance near the square root of the length of each code in the family. Unfortunately, only a small number of families of permutation trinomials with their differential properties over $\mathbb{F}_{2^m}$ are known. Recently, Helleseth, Li and Xia \cite{SETA24} showed that the trinomial $F_{1}(x)=x+x^{3}+x^{2^{(m+1)/2}+1}$ over $\mathbb{F}_{2^m}$, where $m$ is odd integer, is differentially $4$-uniform, known as the $\textit{Welch permutation}$. In Section \ref{3.3}, we have shown that the binary cyclic code $\mathcal{C}_{s}$ gives optimal parameters $[2^m-1,2^m-2-3m,8]$ for $m\geq 5$ is odd when $F_{1}(x)$ is plugged into Eq. (\ref{B}). On the other hand, consider the permutation trinomial $F_{2}(x)=x+x^{2^{(m+2)/2}-1}+x^{2^m-2^{m/2}+1}$, where $m\geq 2$ is even, over $\mathbb{F}_{2^m}$ (see Theorem 3.2 in \cite{SETA2}). Although we have no information about the differential properties of the family of trinomials $F_{2}(x)$, still the binary cyclic code $\mathcal{C}_{s}$ provides optimal parameters $[2^m-1,2^m-1-m,3]$, equivalent to the Hamming code when $F_{2}(x)$ is employed into Eq. (\ref{B}). These facts motivate us to investigate more permutation monomials and trinomials that may yield optimal or near-optimal binary cyclic codes with desirable parameters.} 
\vskip 1pt
\textcolor{black}{Inspired by the ideas and techniques in \cite{SETA3, SETA5, SETA8}, the objective of this paper is to investigate some infinite families of binary cyclic codes using the trace sequence approach by employing several known infinite families of permutation monomials and trinomials over $\mathbb{F}_{2^m}$. The binary cyclic codes presented in this paper have length $v=2^m-1$ with dimensions larger than $v/2$ and minimum distance near $\sqrt{v}$. We determined the upper and lower bounds of these cyclic codes. In Section \ref{3.1} a known infinite family of optimal cyclic code with parameters $[2^m-1,2^m-2-m,4]$, where $m\geq 3$ is odd (see Theorem~1 of \cite{SETA8}) and in Section \ref{3.3} a new infinite family of optimal cyclic code with parameters $[2^m-1,2^m-2-3m,8]$, where $m\geq 5$ is odd, is produced by choosing two suitable permutation trinomials over $\mathbb{F}_{2^m}$. The results of the paper are summarized in Table \ref{Table3}, \ref{Table1} and \ref{table2}.}
\vskip 1pt
The rest of this paper is organized as follows. Section 2 states some essential definitions and related results. Next, in Sections 3 and 4, we study binary cyclic codes more in-depth from different known families of permutation monomials and trinomials over $\mathbb{F}_{2^m}$, for $m$ being odd and even, respectively. Section 5 concludes the paper.

\begin{table}[htbp]
    \centering
    \footnotesize 
    \renewcommand{\arraystretch}{1.5} 
    \setlength{\tabcolsep}{10pt} 
    \begin{adjustbox}{max width=1.2\textwidth, center} 
        \begin{tabular}{|>{\centering\arraybackslash}m{3.8cm}|>{\centering\arraybackslash}m{5.3cm}|>{\centering\arraybackslash}m{5.5cm}|>{\centering\arraybackslash}m{6.2cm}|>{\centering\arraybackslash}m{2.5cm}|}
            \hline
            \textbf{\normalsize $F(x)$} & \textbf{\normalsize Conditions} & \textbf{\normalsize $\dim(\mathcal{C}_{s})$} & \textbf{\normalsize $d(\mathcal{C}_{s})$} & \textbf{\normalsize References} \\
            \hline
            $x^{2^{m}-2}$ & $m\geq 2$ is an integer & $2^{m-1} - 1$ &  $d(\mathcal{C}_{s})^{2}-d(\mathcal{C}_{s})+1\geq v$ such that $d(\mathcal{C}_{s})$ is even and $m$ is odd & Theorem 4.2 in \cite{SETA5} \\
            \hline
            $x^{2^{(m-1)/2}+3}$ & $m \geq 7$ is odd & $2^m - 2 - 5m$ & $d(\mathcal{C}_{s}) \geq 8$ & Theorem 8 in \cite{SETA3} \\
            \hline
            \multirow{2}{*}{$x^{2^{h}-1}$, where $2 \leq h\leq\lceil\frac{m}{2}\rceil$} & $m$ is even & $2^{m}-1-\frac{m(2^{h}+(-1)^{h-1})}{3}$ & $d(\mathcal{C}_{s})\geq 2^{h-2}+1$ & \multirow{2}{*}{Theorem 12 in \cite{SETA3}} \\
            & $m$ is odd & $2^{m}-2-\frac{m(2^{h}+(-1)^{h-1})}{3}$ & $d(\mathcal{C}_{s})\geq 2^{h-2}+2$ and $h>2$ & \\
            \hline
           \multirow{2}{*}{\parbox{3.5cm}{$x^{2^{(m-1)/2}+2^{(m-1)/4}-1}$, where $m\equiv 1\pmod{4}$}} & $m \equiv 1 \pmod{8}$ and $m\geq 9$ & $2^m - 2 - \frac{m(2^{(m+7)/4}+(-1)^{(m-5)/4})}{3}$ & $d(\mathcal{C}_{s})\geq 2^{(m-1)/4}+2$ & \multirow{2}{*}{Theorem 18 in \cite{SETA3}} \\
            & $m \equiv 5 \pmod{8}$ and $m\geq 9$ & $2^m - 2 - \frac{m(2^{(m+7)/4}+(-1)^{(m-5)/4}-6)}{3}$ & $d(\mathcal{C}_{s})\geq 2^{(m-1)/4}$ & \\
            \hline
            \multirow{2}{*}{\parbox{3.5cm}{$x^{2^{2h}-2^{h}+1}$,\text{  }where $\gcd(m,h)=1$}} & 
            $1\leq h \leq \begin{cases}
                \frac{m-1}{4} & m \equiv 1 \pmod{4} \\
                \frac{m-3}{4} & m \equiv 3 \pmod{4} \\
                \frac{m-4}{4} & m \equiv 0 \pmod{4} \\
                \frac{m-2}{4} & m \equiv 2 \pmod{4}
            \end{cases}$ & 
             $\begin{cases} 
                2^{m}-1 - \frac{m(2^{h+2}+(-1)^{h-1})+3\mathbb{N}_{2}(m)}{3}; &  \\
                \text{if } h \text{ is even}, & \\
                2^{m}-1 - \frac{m(2^{h+2}+(-1)^{h-1}-6)+3\mathbb{N}_{2}(m)}{3}; & \\
                \text{if } h \text{ is odd}
            \end{cases}$ & 
            $d(\mathcal{C}_{s})\geq\begin{cases}
                2^{h}+2; &\text{if } h \text{ is even and }m\text{ is odd} \\      
                2^{h}+1; & \text{if } h \text{ is even and }m\text{ is even} \\
                2^{h}; & \text{if } h \text{ is odd}
            \end{cases}$ & Theorem 19 in \cite{SETA12} \\
            \hline
            \multirow{2}{*}{$x^{2^{h}+1}$} & $m$ is odd and $\operatorname{gcd}(m,h)=1$ & $2^{m}-2-m$ & 4 & \multirow{2}{*}{Theorem 1 in \cite{SETA8}}\\
            & $m \equiv 2 \pmod{4}$ and $\operatorname{gcd}(m,h)=2$ & $2^{m}-1-m$ & 3 & \\
            \hline
            $x^{2^{(m-1)/2}+2^{(3m-1)/4}-1}$  & $m \equiv 3 \pmod{4}$ and $m\geq 3$ & $2^{m} - 1-\textit{L}_{s}$, $\textit{L}_{s}$ given in Corollary 4 of \cite{SETA9} &  $d(\mathcal{C}_{s})\geq 2^{(m+1)/4}+2$  & Theorem 5 in \cite{SETA9} \\
            \hline
            \multirow{2}{*}{\parbox{3.5cm}{$x^{2^{4h}+2^{3h}+2^{2h}+2^{h}-1}$, where $m=5h$}} & $h$ is even & $2^{m}-1- m\left(\frac{22}{3}(2^{h}-1)-3h\right)$ & $d(\mathcal{C}_{s})\geq 2^{h}+1$ & \multirow{2}{*}{Theorem 6 in \cite{SETA11}}\\
            & $h$ is odd & $2^{m}-2- m\left(\frac{22}{3}(2^{h}-2)-3h+6\right)$ & $d(\mathcal{C}_{s})\geq 2^{h}+2$ & \\
            \hline
            $x^{2^{2h}+2^{h}+1}$, where $m=4h$  & $h$ is odd & $2^{m} - 1-\frac{5m}{2}$ &  3  & Theorem 6 in \cite{SETA8} \\
            \hline
             \multirow{2}{*}{$x^{2^{2h}-2^{h}+1}$} & $\frac{m}{3}\geq h>\begin{cases}
                 \frac{m-1}{4},\text{ if $m\equiv 1 \pmod{4}$} \\
                 \frac{m-2}{4},\text{ if $m\equiv 2 \pmod{4}$} \\
                 \frac{m-3}{4},\text{ if $m\equiv 3\pmod{4}$} \\
                 \frac{m-4}{4},\text{ if $m\equiv 0 \pmod{4}$} 
             \end{cases}$ & $2^{m}-1- \textit{L}_{s}$, $\textit{L}_{s}$ given in Lemma 11 of \cite{SETA8} & $d(\mathcal{C}_{s})\geq \begin{cases}
                 2^{m-3h+1};\text{ if $h$ is odd} \\
                 2^{m-3h+1}+1;\text{ if $h$ is even and $m$ is even} \\
                 2^{m-3h+1}+2;\text{ if $h$ is even and $m$ is odd} \\
             \end{cases}$ & Theorem 4 in \cite{SETA8}\\
            & $\frac{2m+3}{5}>h>\frac{m}{3}$ & $2^{m}-1- \textit{L}_{s}$, $\textit{L}_{s}$ given in Lemma 12 of \cite{SETA8} & $d(\mathcal{C}_{s})\geq 2^{h}-2^{m-2h}$ & Theorem 5 in \cite{SETA8} \\
            \hline
            \multirow{2}{*}{$x^{2^{2h}-2^{h}+1}$,\text{ }where\text{ }$h=\frac{m-1}{2}$} & $m \equiv 1 \pmod{4}$ and $m \geq 9$ & $(2^{(m-1)/2} - 1)(2^{(m+1)/2} - m + 2)$ & $2^{(m-3)/2}\leq d(\mathcal{C}_{s})\leq 1+m(2^{(m-1)/2}-1)$ & \multirow{2}{*}{Theorem 5}\\
            & $m \equiv 3 \pmod{4}$ and $m \geq 7$ & $(2^{(m-1)/2} - 1)(2^{(m+1)/2} - m + 2) + 2m$ & $2^{(m-3)/2}\leq d(\mathcal{C}_{s})\leq 1+m(2^{(m-1)/2}-3)$ & \\
            \hline
        \end{tabular}
    \end{adjustbox}
    \caption{Known binary cyclic codes $\mathcal{C}_{s}$ from monomials $F(x)$ over $\mathbb{F}_{2^{m}}$ with parameters $[2^{m}-1, \dim(\mathcal{C}_{s}), d(\mathcal{C}_{s})]$.}
    \label{Table3}
\end{table}

\begin{table}[htbp]
    \centering
    \footnotesize 
    \renewcommand{\arraystretch}{1.5} 
    \setlength{\tabcolsep}{10pt} 
    \begin{adjustbox}{max width=1.2\textwidth, center} 
        \begin{tabular}{|>{\centering\arraybackslash}m{4cm}|>{\centering\arraybackslash}m{4cm}|>{\centering\arraybackslash}m{4.6cm}|>{\centering\arraybackslash}m{9.8cm}|>{\centering\arraybackslash}m{2.5cm}|}
            \hline
            \textbf{\normalsize $F(x)$} & \textbf{\normalsize Conditions} & \textbf{\normalsize $\dim(\mathcal{C}_{s})$} & \textbf{\normalsize $d(\mathcal{C}_{s})$} & \textbf{\normalsize References} \\
            \hline
            \multirow{3}{*}{\parbox{3.5cm}{$x+x^{r}+x^{2^{h}-1}$, where $\operatorname{w}_{2}(r)=m-1$ and $0 \leq h\leq\lceil\frac{m}{2}\rceil$}} & $m$ is odd and $h=0$ & $2^{m-1}-1-m$ & $d(\mathcal{C}_{s})\geq 8$ & \multirow{2}{*}{Theorem 26 in \cite{SETA3}} \\
            & $m$ is even and $h=0$ & $2^{m-1}-1+m$ & $d(\mathcal{C}_{s})\geq 3$ & \\
            & $h\neq 0$  & $2^{m-1}-1$ & $d(\mathcal{C}_{s})\geq \begin{cases}
                2^{(m-1)/2}+4,\text{ if $m\equiv 1\pmod{4}$, $m\geq 5$ and $0< h\leq\frac{m-3}{2}$} \\
                2^{(m-1)/2}+4,\text{ if $m\equiv 3\pmod{4}$, $m\geq 7$, $0< h\leq\frac{m-3}{2}$ and $2\mid h$} \\
                2^{(m-1)/2}+2,\text{ if $m\equiv 3\pmod{4}$, $m\geq 7$, $0< h\leq\frac{m-3}{2}$ and $2\nmid h$} \\
            \end{cases}$ & Theorem 6, 7 in \cite{SETA22}\\
            \hline
            $x+x^{2^{(m+2)/2}-1}+x^{2^m-2^{m/2}+1}$ & $m \geq 2$ is even & $2^m - 1 - m$ & 3 & This paper \\
            \hline
          $x + x^{2^{(m+1)/2}-1} + x^{2^m-2^{(m+1)/2}+1}$ & $m \geq 3$ is odd & $2^m - 2 - m$ & 4 & Theorem 1\\
            \hline
            $x^{3\cdot2^{(m+1)/2}+4} + x^{2^{(m+1)/2}+2} + x^{2^{(m+1)/2}}$ & $m \geq 7$ is odd & $2^m - 2 - 3m$ & $4 \leq d(\mathcal{C}_{s}) \leq 8$ & Theorem 2 \\
            \hline
            $x + x^3 + x^{2^{(m+1)/2}+1}$ & $m \geq 5$ is odd & $2^m - 2 - 3m$ & 8 & Theorem 3 \\
            \hline
            \multirow{2}{*}{$x + x^3 + x^{2^m - 2^{(m+3)/2} + 2}$} & $m \equiv 1 \, (\text{mod } 4)$ and $m \geq 5$ & $2(2^{m-1}-1) - m(2^{(m-3)/2} + 1)$ & $\max\{8, 2^{(m-5)/2}+2\}\leq d(\mathcal{C}_{s})\leq 1+m(2^{(m-3)/2}+1)$ & \multirow{2}{*}{Theorem 4}\\
	 & $m \equiv 3 \, (\text{mod } 4)$ and $m \geq 7$ & $2(2^{m-1}-1) - m(2^{(m-3)/2} - 1)$ & $2^{(m-5)/2} + 2 \leq d(\mathcal{C}_{s})\leq 1+m(2^{(m-3)/2}-1)$ & \\
            \hline
            \multirow{2}{*}{$x + x^{2^{m/2}} + x^{2^m - 2^{m/2} + 1}$} & $m \equiv 0 \, (\text{mod } 4)$ and $m \geq 4$ & $2^m - 1 - 2m \cdot \left(\frac{2^{m/2}-1}{3}\right)$ & $2^{(m-2)/2}\leq d(\mathcal{C}_{s})\leq 1+m\left(\frac{2^{\frac{m}{2}+1}-2}{3}\right)$ & \multirow{2}{*}{Theorem 6}\\
			 & $m \equiv 2 \, (\text{mod } 4)$ and $m \geq 6$ & $2^m - 1 - 4m \cdot \left(\frac{2^{m/2-1}-1}{3}\right)$ & $2^{(m-2)/2}\leq d(\mathcal{C}_{s})\leq 1+m\left(\frac{2^{\frac{m}{2}+1}-4}{3}\right)$ & \\
            \hline
            $x + x^{2^{m/2+1}-1} + x^{2^m - 2^{m/2+1} + 2}$ & $m \geq 6$ is even & $2^m - 1 - m(2^{(m-2)/2} - 1)$ & $2^{(m-4)/2} + 1\leq d(\mathcal{C}_{s})\leq 1+m(2^{(m-2)/2}-1)$ & Theorem 7 \\
            \hline
             \multirow{2}{*}{$x + x^{2^{m/2}} + x^{2^{m-1} - 2^{m/2-1} + 1}$} & $m \equiv 0 \, (\text{mod } 4)$ and $m \geq 8$ & $2^{m}-1 - m\cdot 2^{m/2}-\frac{m}{2}$ & $\max\{7, 2^{(m-2)/2}+1\}\leq d(\mathcal{C}_{s})\leq 1+m(2^{m/2}+1)-\frac{m}{2}$ & \multirow{2}{*}{Theorem 8}\\
 & $m \equiv 2 \, (\text{mod } 4)$ and $m > 6$ & $2^{m}-1 - m\cdot 2^{m/2}+\frac{3m}{2}$ & $2^{(m-2)/2} + 1\leq d(\mathcal{C}_{s})\leq 1+m(2^{m/2}-1)-\frac{m}{2}$ & \\
            \hline
        \end{tabular}
    \end{adjustbox}
    \caption{Known binary cyclic codes $\mathcal{C}_{s}$ from trinomials $F(x)$ over $\mathbb{F}_{2^{m}}$ with parameters $[2^{m}-1, \dim(\mathcal{C}_{s}), d(\mathcal{C}_{s})]$.}
    \label{Table1}
\end{table}


\section{Preliminaries}

Throughout this paper, we set $v=2^m-1$. In this section, first we state some essential results related to $2$-cyclotomic cosets modulo $v$, then we discuss a well-known approach of designing cyclic codes by periodic sequences. We require all these ingredients in the subsequent sections. 
\subsection{Essential results on 2-cyclotomic cosets modulo $v$}
Let $\mathbb{Z}_{v}=\{0,1,2,\dots,v-1\}$. For any $i\in\mathbb{Z}_{v}$, the $2$-cyclotomic coset $C_i$ of $i$ modulo $v$ is defined as
\begin{equation*}
	C_i=\{2\cdot i^{s}:\text{ $0\leq s\leq \ell_{i}-1$}\}\text{(mod $v$)},
\end{equation*} where $\ell_i$ is the least positive integer such that $i\equiv 2^{\ell_i}\cdot i$ (mod $v$), and is the size of $C_{i}$. The size of a $2$-cyclotomic coset modulo $v$ divides $m$. The least integer in $C_i$ is called the coset leader of $C_i$. We use the notation $\Gamma$ to denote the set of all coset leaders. Let $\alpha$ be the primitive element of $\mathbb{F}_{2^m}$, and let $m_{\alpha^i}(x)$ denote the minimal polynomial of $\alpha^{i}$ over $\mathbb{F}_{2}$. We know that
\begin{align*}
	\hskip 15pt	\bigcup_{i\in \Gamma}C_i &= \mathbb{Z}_{v} \text{, }  m_{\alpha^i}(x)=\prod_{s\in C_{i}}(x-\alpha^{s}) \text{ and } x^v-1=\prod_{i\in\Gamma}m_{\alpha^i}(x).  &
\end{align*} 
The $2$-adic expansion of an integer $i$ with $0 \leq i \leq 2^m-1$, is defined as 
\begin{align*}\label{EE}
	i&=i_0+i_{1}\cdot2+\dots+i_{m-1}\cdot2^{m-1}, 
\end{align*}
where $i_{0},i_{1},\dots,i_{m-1}\in \{0,1\}$. Define $\operatorname{w}_{2}(i)=\sum_{j=0}^{m-1}i_{j}$, and we call it the $2$-weight of $i$ in the sequel.
\vskip 1pt
We need the following lemmas in the subsequent sections.
\vspace{0.5em}
\begin{lemma}[\cite{SETA7}]
	\label{L5}
	For any coset leader \( i \in \Gamma \setminus \{0\} \), \( i \) is odd and \( 1 \leq i < 2^{m-1} \).
\end{lemma}
\vspace{0.5em}
\begin{lemma}[\cite{SETA3}]
	\label{lem:L1}
	Let \( n = \lceil \frac{m+1}{2} \rceil \) and \( \Gamma' = \{ 1 \leq i \leq 2^n - 1 : \textnormal{$i$ is an odd integer} \} \). Then, for any \( i \in \Gamma' \), we have:
	\begin{enumerate}[(i)]
		\item \( i \) is the coset leader of \( C_i \);
		\item \( \ell_i = m \), except that \( \ell_{2^{\frac{m}{2}}+1} = \frac{m}{2} \) for even \( m \).
	\end{enumerate}
\end{lemma}
\vspace{0.5em}
\begin{remark}
	From Lemma \ref{lem:L1}, we conclude that \( C_{i} \cap C_{j} = \emptyset \) for any distinct \( i, j \in \Gamma' \).
\end{remark}
\subsection{Cyclic codes designed by periodic sequences} 
Let $s^{\infty}=(s_{i})_{i=0}^{\infty}$ be a sequence of period $v$ over $\mathbb{F}_{2}$. The polynomial $M(x)=1+m_{1}x+m_{2}x^{2}+\dots+m_{l}x^{l}$ over $\mathbb{F}_{2}$ is called the \textit{minimal polynomial} of $s^{\infty}$ if $l$ is the smallest positive integer such that 
\begin{flalign*}
	\hskip 15pt	-s_{i}&=m_{1}s_{i-1}+m_{2}s_{i-2}+\dots+m_{l}s_{i-l} \text{ for all $i\geq l$.} &
\end{flalign*}
Throughout this paper, the minimal polynomial of the sequence $s^{\infty}$ is denoted by the notation $\textit{g}_{s}(x)$. The degree of the polynomial $\textit{g}_{s}(x)$ is known as the \textit{linear span} of $s^{\infty}$ and we denote it by the notation $L_{s}$. The cyclic code with generator polynomial $\textit{g}_{s}(x)$ is referred to as $\mathcal{C}_{s}$, and we call the cyclic code $\mathcal{C}_{s}$ as the code designed by the sequence $s^{\infty}$.
\vskip 1pt
The following well-known Lemma \cite{SETA1} provides an efficient way to determine the generator polynomial $\textit{g}_{s}(x)$ and the linear span $L_{s}$ corresponding to any sequence $s^{\infty}$ of period $v$.
\vspace{0.5em}
\begin{lemma}\label{C}
	\label{lem:technical}
	For any sequence $s^{\infty}=(s_{t})_{t=0}^{\infty}$ over $\mathbb{F}_{2}$ of period $2^m-1$, the component $s_{t}$ has a unique expansion of the form
	\begin{equation*}\label{A}
		s_t=\sum_{i=0}^{2^m-2}a_i\alpha^{it}\hspace{1mm} \textnormal{ for all } t\geq 0,
	\end{equation*}
	where $a_{i}\in\mathbb{F}_{2^m}$. Let the index set be $I_{s}=\{i: a_i\neq 0\}$, then the minimal polynomial $\textit{g}_{s}(x)$ of $s^{\infty}$ is $\prod_{j\in I_{s}}^{}(1-\alpha^j x)$, and the linear span of $s^{\infty}$ is $L_{s}=|I_{s}|$.
\end{lemma}

\vspace{0.5em} 

\begin{remark}
	From the above discussion, we conclude that the generator polynomial of the cyclic code $\mathcal{C}_{s}$ is given by
	\begin{align*}
		\textit{g}_{s}(x)=\prod_{i\in I_{s}\cap\Gamma}m_{\alpha^{-i}}(x)
	\end{align*}
\end{remark}

\begin{table}[htbp]
    \centering
    \caption{Optimality of $\mathcal{C}_{s}$ and $\mathcal{C}_{s}^{\perp}$ from polynomials $f_{i}(x)$ over $\mathbb{F}_{2^m}$}
    \label{table2}
    \footnotesize
    \renewcommand{\arraystretch}{1.2}
    \setlength{\tabcolsep}{4pt}
    \begin{tabular}{|c|c|l|l|l|l|}
        \hline
        \textbf{i} & \textbf{m} & \textbf{$\mathcal{C}_{s}$} & \textbf{$\mathcal{C}_{s}^{\perp}$} & \textbf{Optimality of $\mathcal{C}_{s}$} & \textbf{Optimality of $\mathcal{C}_{s}^{\perp}$} \\ 
        \hline
        1 & Any odd $\geq3$ & $[2^m-1,2^m-2-m,4]$ & -- & Optimal family & -- \\ 
        \hline
        2 & 5 & $[31,25,4]$ & $[31,6,15]$ & Optimal & Optimal \\
        \hline
        2 & 7 & $[127,105,6]$ & $[127,22,43]$ & No & No \\
        \hline
        3 & Any odd $\geq5$ & $[2^m-1,2^m-2-3m,8]$ & -- & Optimal family & -- \\ 
        \hline
        4 & 5 & $[31,15,8]$ & $[31,16,7]$ & Optimal & Near optimal \\ 
        \hline
        4 & 7 & $[127,105,6]$ & $[127,22,43]$ & No & No \\
        \hline
        5 & 5 & $[31,15,8]$ & $[31,16,7]$ & Optimal & Near optimal \\
        \hline
        5 & 7 & $[127,91,8]$ & $[127,36,28]$ & No & No \\
        \hline
        6 & 4 & $[15,7,3]$ & $[15,8,4]$ & No & Optimal \\
        \hline
        6 & 6 & $[63,39,7]$ & $[63,24,12]$ & No & No \\
        \hline
        6 & 8 & $[255,175,15{\leq}d(\mathcal{C}_{s}){\leq}17]$ & $[255,80,40]$ & No & No \\
        \hline
        7 & 4 & $[15,11,3]$ & $[15,4,8]$ & Optimal & Optimal \\
        \hline
        7 & 6 & $[63,45,5]$ & $[63,18,16]$ & No & No \\
        \hline
        7 & 8 & $[255,199,10]$ & $[255,56,64]$ & No & No \\
        \hline
        8 & 6 & $[63,28,9]$ & $[63,35,10]$ & No & No \\
        \hline
        8 & 8 & $[255,123,20{\leq}d(\mathcal{C}_{s}){\leq}31]$ & $[255,132,22{\leq}d(\mathcal{C}_{s}){\leq}24]$ & No & No \\
        \hline
    \end{tabular}
    \textit{\char`# \hspace{0.5em}Near optimal means 1 smaller than the best minimum distance in $\cite{SETA23}$. The computation of the minimum distances for the infinite families $\mathcal{C}_{s}^{\perp}$ when $i=1,3$ is still an open problem.}
\end{table}
\section{Binary cyclic codes from polynomials over $\mathbb{F}_{2^m}$, $m$ is odd}
\subsection{Binary code $\mathcal{C}_{s}$ from the trinomial $x + x^{2^{(m+1)/2}-1}+x^{2^m-2^{(m+1)/2}+1} $}\label{3.1}

Let us consider the permutation trinomial $f_{1}(x)= x + x^{2^{(m+1)/2}-1}+x^{2^m-2^{(m+1)/2}+1} $ over $\mathbb{F}_{2^m}$, where $m$ is an odd integer $(\text{see Theorem 2.1 of } \cite{SETA2})$. This subsection studies the binary cyclic code $\mathcal{C}_{s}$ designed by the sequence defined in Eq. (\ref{B}) from $f_{1}(x)$ over $\mathbb{F}_{2^m}$. We now prove the following result.
\vspace{0.5em}
\begin{theorem}\label{E}
	Let $m = 2h+1 \geq 3$ and $s^{\infty}$ be the sequence defined in Eq. $(\ref{B})$ from the trinomial $f_{1}(x)$ over $\mathbb{F}_{2^m}$. Then the binary cyclic code $\mathcal{C}_{s}$ has parameters $[2^{m}-1, 2^{m}-2-m, 4]$ with the generator polynomial given by
	\begin{equation*}
		\textit{g}_{s}(x) = (x-1)m_{\alpha^{-1}}(x).
	\end{equation*}
\end{theorem}
\begin{proof}
	For $m$ being odd, $\operatorname{Tr}\left(1\right) = 1$. From Eq. (\ref{B}), we have
	\begin{align}\label{EE1}
		s_{t} &= \operatorname{Tr}\left(f_{1}(\alpha^t+1)\right) \nonumber\\
		&= \operatorname{Tr}\left((\alpha^t+1) + (\alpha^t+1)^{2^{h+1}-1} + (\alpha^t+1)^{1+2^{h+1}\sum_{i=0}^{h-1}2^{i}}\right) \nonumber\\
		&= \operatorname{Tr}\left((\alpha^t+1) + (\alpha^t+1)^{\sum_{i=0}^{h}2^{i}} + (\alpha^t+1)(\alpha^{t}+1)^{\sum_{i=0}^{h-1}2^{h+1+i}}\right) \nonumber\\
		&= \operatorname{Tr}\left((\alpha^t+1) + \sum_{i=0}^{2^{h+1}-1}(\alpha^t)^{i} + (\alpha^t+1)\sum_{i=0}^{2^{h}-1}(\alpha^t)^{i \cdot 2^{h+1}}\right) \nonumber\\
		&= \operatorname{Tr}\left((\alpha^{t}+1) + \sum_{i=0}^{2^{h+1}-1}(\alpha^{t})^{i} + \sum_{i=0}^{2^{h}-1}(\alpha^{t})^{i \cdot 2^{h+1}+1} + \sum_{i=0}^{2^{h}-1}(\alpha^{t})^{i}\right) \nonumber\\
		&= \operatorname{Tr}\left((\alpha^{t}+1) + \sum_{i=2^h}^{2^{h+1}-1}(\alpha^{t})^{i} + \sum_{i=0}^{2^{h}-1}(\alpha^{t})^{i+2^h}\right) \nonumber\\
		&= \operatorname{Tr}\left(\alpha^t\right) + 1.
	\end{align}
	
	The $2$-cyclotomic coset $C_{1}$ is of size $m$. From Eq. (\ref{EE1}), we have $s_{t}=1+\sum_{i\in C_{1}}(\alpha^t)^{i}\text{ for all }t\geq 0$. The index set $I_{s}$ corresponding to the sequence $s^{\infty}$ of (\ref{EE1}) is $C_{1}\cup\{0\}$, and the linear span $L_{s}$ of $s^{\infty}$ is $|I_{s}|=m+1$. As $0$ and $1$ are the only coset leaders in $I_{s}$, the results on the dimension of the code $\mathcal{C}_{s}$ and its generator polynomial follow directly from Lemma \ref{C}.

	Let $d(\mathcal{C}_{s})$ denote the minimum distance of the code $\mathcal{C}_{s}$. The reciprocal of the generator polynomial $\textit{g}_{s}(x)$ has roots $1$, $\alpha$, and $\alpha^{2}$. As we know, the code $\mathcal{C}_{s}$ and the code generated by the reciprocal of $\textit{g}_{s}(x)$ both have the same weight distribution. Hence, $d(\mathcal{C}_{s}) \geq 4$ from the BCH bound. From the dimension of $\mathcal{C}_{s}$ and by the sphere-packing bound, we obtain $d(\mathcal{C}_{s}) \leq 4$. Therefore, $d(\mathcal{C}_{s}) = 4$.
\end{proof}

\subsection{Binary code $\mathcal{C}_{s}$ from the trinomial $x^{3\cdot2^{(m+1)/2}+4}+x^{2^{(m+1)/2}+2}+x^{2^{(m+1)/2}}$}
Let $f_{2}(x)=x^{3\cdot2^{(m+1)/2}+4}+x^{2^{(m+1)/2}+2}+x^{2^{(m+1)/2}}$ over $\mathbb{F}_{2^m}$, where $m$ is an odd integer. This subsection deals with the cyclic code $\mathcal{C}_{s}$ from the permutation trinomial $f_{2}(x)$ over $\mathbb{F}_{2^m}$ (see \cite{SETA16} or Theorem $4$ in \cite{SETA17}).
\begin{theorem}\label{Th2}
	Let $m=2h+1 \geq 7$ and $s^{\infty}$ be the sequence defined in Eq. \textnormal{(\ref{B})} from the trinomial $f_{2}(x)$ over $\mathbb{F}_{2^m}$. Then the binary cyclic code $\mathcal{C}_{s}$ has parameters $[2^m-1, 2^m-2-3m,d(\mathcal{C}_{s})]$, where $4 \leq d(\mathcal{C}_{s}) \leq 8$, with the generator polynomial given by
	\begin{align*}
		\textit{g}_{s}(x) &= (x-1) m_{\alpha^{-3}}(x) m_{\alpha^{-1-2^{h-1}}}(x) m_{\alpha^{-1-2^{h-1}-2^{h}}}(x).
	\end{align*}
	\begin{proof}
		We know that $\operatorname{Tr}(x^{2^{t}}) = \operatorname{Tr}(x)$ for any integer $t \geq 0$ and $x \in \mathbb{F}_{2^m}$. By definition, we have
		\begin{align}\label{E26}
			s_{t} &= \operatorname{Tr}\left(f_{2}(\alpha^{t}+1)\right) \nonumber \\
			&= \operatorname{Tr}\left((\alpha^{t}+1)^{2^{h}+2^{h-1}+1} + (\alpha^{t}+1)^{2^{h}+1} + (\alpha^{t}+1)\right) \nonumber \\
			&= \operatorname{Tr}\left((\alpha^{t})^{2^{h}+2^{h-1}+1} + (\alpha^{t})^{2^{h-1}+1} + (\alpha^{t})^{3}\right) + 1.
		\end{align}
		By Lemma \ref{lem:L1}, we know that the $2$-cyclotomic cosets $C_{3}$, $C_{2^{h-1}+1}$, and $C_{2^{h}+2^{h-1}+1}$ are of size $m$, and their coset leaders are $3$, $2^{h-1}+1$, and $2^{h}+2^{h-1}+1$ respectively. Hence, they are pairwise disjoint. From Eq. (\ref{E26}), we have $s_{t}=1+\sum_{i\in C_{3}}(\alpha^t)^{i}+\sum_{i\in C_{2^{h-1}+1}}(\alpha^t)^{i}+\sum_{i\in C_{2^{h}+2^{h-1}+1}}(\alpha^t)^{i}\text{ for all }t\geq 0$. The index set $I_{s}$ corresponding to the sequence $s^{\infty}$ of (\ref{E26}) is $\{0\}\cup C_{3}\cup C_{2^{h-1}+1}\cup C_{2^{h}+2^{h-1}+1}$, and the linear span $L_{s}$ of $s^{\infty}$ is $|I_{s}|=3m+1$. The results on the dimension of the code $\mathcal{C}_{s}$ and its generator polynomial follow directly from Lemma \ref{lem:technical}.
        \vskip 1pt
        Note that $\mathcal{C}_{s}$ is an even-weight code, and the reciprocal of $\textit{g}_{s}(x)$ has the roots $\alpha^{2^{h}+2^{h-1}}$ and $\alpha^{2^{h}+2^{h-1}+1}$. By the BCH bound and the sphere-packing bound, we conclude that $4 \leq d(\mathcal{C}_{s}) \leq 8$. 
	\end{proof}
\end{theorem}

\begin{example}\label{EX2}
	Let $m=5$ and $\alpha$ be a root of the primitive polynomial $x^5 + x^2 + 1$ over $\mathbb{F}_2$. The generator polynomial of $\mathcal{C}_s$ is $\textit{g}_s(x) = x^6 + x^2 + x + 1$. Then, $\mathcal{C}_s$ is a $[31, 25, 4]$ binary cyclic code and $\mathcal{C}_s^{\perp}$ is a $[31, 6, 15]$ binary cyclic code. According to the database \cite{SETA23}, both $\mathcal{C}_s$ and $\mathcal{C}_s^{\perp}$ are optimal.
\end{example}
\vspace{0.5em}
\begin{example}
		Let $m=7$ and $\alpha$ is a root of the primitive polynomial $x^7+x+1$ over $\mathbb{F}_{2}$. The minimal polynomial of $s^{\infty}$ is $\mathbb{M}_{s}(x)=x^{22}+x^{21}+x^{20}+x^{18}+x^{16}+x^{15}+x^{13}+x^{12}+x^{11}+x^{10}+x^8+x^7+x^5+x^3+x^2+1$. Then $\mathcal{C}_{s}$ is a binary $[127, 105, 6]$ cyclic code and its dual $\mathcal{C}_{s}^{\perp}$ is a $[127,22,43]$ cyclic code.
	\end{example}
    
\subsection{Binary code $\mathcal{C}_{s}$ from the trinomial $x+x^3+x^{2^{(m+1)/2}+1}$}\label{3.3}
In 1999, Dobbertin \cite{SETA4} showed the bijectivity of the polynomial $f_{3}(x)=x+x^3+x^{2^{(m+1)/2}+1}$ over $\mathbb{F}_{2^{m}}$, where $m$ is an odd integer. This subsection focuses on the binary code $\mathcal{C}_s$ from the permutation trinomial $f_{3}(x)$ over $\mathbb{F}_{2^{m}}$. 
\vspace{0.5em}
\begin{theorem}\label{L2}
	Let $m = 2h + 1 \geq 5$ and $s^{\infty}$ be the sequence defined in Eq. $(\ref{B})$ from the trinomial $f_3(x)$ over $\mathbb{F}_{2^m}$. Then, the binary cyclic code $\mathcal{C}_s$ has parameters $[2^m-1, 2^m-2-3m, 8]$ with the generator polynomial given by
	\begin{flalign}\label{E1}
		\hskip 30pt \textit{g}_s(x) &= (x-1) m_{\alpha^{-1}}(x) m_{\alpha^{-3}}(x) m_{\alpha^{-(2^h+1)}}(x). &
	\end{flalign}
	
	\begin{proof}
		We know that $\operatorname{Tr}(x^{2^h}) = \operatorname{Tr}(x)$ and $x^{2^{2h+1}} = x$ for all $x \in \mathbb{F}_{2^{2h+1}}$. By definition, we have
		\begin{flalign}\label{G}
			\hskip 30pt s_t &= \operatorname{Tr}\left(f_3(\alpha^t + 1)\right) \nonumber \\
			&= \operatorname{Tr}\left((\alpha^t + 1) + (\alpha^t + 1)^3 + (\alpha^t + 1)^{2^{h+1} + 1}\right) \nonumber \\
			&= \operatorname{Tr}\left((\alpha^t)^{2^{h+1} + 1} + (\alpha^t)^3 + (\alpha^t) + 1\right) \nonumber \\
			&= \operatorname{Tr}\left((\alpha^t)^{2^h + 1} + (\alpha^t)^3 + \alpha^t\right) + 1. &
		\end{flalign}
		
		By Lemma \ref{lem:L1}, we know that the $2$-cyclotomic cosets $C_1$, $C_3$, and $C_{2^h + 1}$ are of size $m$ and are pairwise disjoint. With the help of Lemma \ref{C} and Eq. (\ref{G}), the results on the dimension of the code $\mathcal{C}_s$ and its generator polynomial follow similarly to Theorem \ref{Th2}.
		
		Let $\mathcal{A}_s$ be the cyclic code with the generator polynomial $m_{\alpha^{-1}}(x) m_{\alpha^{-(2^k + 1)}}(x) m_{\alpha^{-(2^{2k} + 1)}}(x)$, where $k = h + 1$. Then, $\mathcal{A}_s$ is a triple-error-correcting code with minimal distance equal to $7$, as $\operatorname{gcd}(k, m) = 1$ \cite{SETA6} or Theorem 1 in \cite{SETA13}. Hence, $\mathcal{C}_s$ is the even-weight subcode of $\mathcal{A}_s$. From the sphere-packing bound, the upper bound of the minimum distance of $\mathcal{C}_{s}$ is $8$. By combining these facts, we get the desired conclusion. 
	\end{proof}
	
\end{theorem}
\subsection{Binary code $\mathcal{C}_{s}$ from the trinomial $x+x^3+x^{2^m-2^{(m+3)/2}+2}$}
Define $f_{4}(x)=x+x^3+x^{2^m-2^{(m+3)/2}+2}$ over $\mathbb{F}_{2^m}$, where $m$ is odd. In Theorem 2.3 of \cite{SETA2}, $f_{4}(x)$ is proved to be a permutation over $\mathbb{F}_{2^m}$. This subsection concentrates on studying binary code $\mathcal{C}_{s}$ from the trinomial $f_{4}(x)$ over $\mathbb{F}_{2^m}$. Let $h=\frac{m-1}{2}$. Then we have
\begin{flalign}
	\hskip 15pt \operatorname{Tr}(f_4(x+1)) &= \operatorname{Tr}\left( (x+1) + (x+1)^3 + (x^2+1)(x^{2^{h+2}}+1)^{2^{h-1}-1} \right) \nonumber & \\
	&= \operatorname{Tr}\left(x^2 + x^3 + (x^2 + 1) \sum_{i=0}^{2^{h-1}-1} x^{i \cdot 2^{h+2}} \right) \nonumber & \\
	&= 1 + \operatorname{Tr}\left(x^3 + \sum_{i=1}^{2^{h-1}-1} x^{i+2^h} + \sum_{i=1}^{2^{h-1}-1} x^{i} \right) 
\end{flalign}

The sequence $s^{\infty}$ of $(\ref{B})$ designed from the trinomial $f_4(x)$ is given by
\begin{flalign}\label{E2}
	\hskip 15pt s_t &= 1 + \operatorname{Tr}\left( (\alpha^t)^3 + \sum_{i=1}^{2^{h-1}-1} (\alpha^t)^{i+2^h} + \sum_{i=1}^{2^{h-1}-1} (\alpha^t)^{i} \right), \text{ for all } t \geq 0. &
\end{flalign}
First, we shall follow some notations given in \cite{SETA3} and present some Lemmas, which will be utilized to determine the generator polynomial of the code $\mathcal{C}_{s}$.\\
Let $t$ be a positive integer. For all odd integers $j \in \{1, 2, 3, \dots, 2^{t}-1\}$, define
\vskip 5pt
\begin{math}
	\epsilon_{j}^{(t)} = 
	\begin{cases}
		1, & \text{if } j = 2^t - 1 \\
		\lceil \log_{2} \left(\frac{2^t - 1}{j}\right) \rceil, & \text{if } 1 \leq j < 2^t - 1
	\end{cases}
\end{math}
\vskip 3pt
and
\begin{math}
	\kappa_{j}^{(t)} = \epsilon_{j}^{(t)} \pmod{2}.
\end{math}
\vskip 1pt
Let $B_{j}^{(t)} = \{2^i j : i = 0, 1, 2, \dots, \epsilon_{j}^{(t)} - 1\}$.
\vskip 5pt
In addition, it is not difficult to verify that
\begin{flalign*}
	\hskip 15pt \bigcup_{1 \leq 2i+1 \leq 2^{t} - 1} B_{2i+1}^{(t)} &= \{1, 2, 3, \dots, 2^t - 1\} \text{ and } B_{j_1}^{(t)} \cap B_{j_2}^{(t)} = \emptyset &
\end{flalign*}
for any distinct pair of odd integers $j_1$ and $j_2$ in $\{1, 2, 3, \dots, 2^t - 1\}$.
\vspace{0.5em}
\begin{lemma} \textnormal{(\cite{SETA3})} \label{L10}
	Let $j$ be an odd integer in $\{1, 2, 3, \dots, 2^{t+1} - 1\}$. Then
	\vskip 1pt
	\begin{itemize}
		\item $B_{j}^{(t+1)} = B_{j}^{(t)} \cup \{j 2^{\epsilon_{j}^{(t)}}\}$ if $1 \leq j \leq 2^t - 1$,
		\item $B_{j}^{(t+1)} = \{j\}$ if $2^t + 1 \leq j \leq 2^{t+1} - 1$,
		\item $\epsilon_{j}^{(t+1)} = \epsilon_{j}^{(t)} + 1$ if $1 \leq j \leq 2^t - 1$,
		\item $\epsilon_{j}^{(t+1)} = 1$ if $2^t + 1 \leq j \leq 2^{t+1} - 1$.
	\end{itemize}
\end{lemma}

\begin{lemma}\label{L4}
	Let $j$ be an odd integer in $\{1, 2, 3, \dots, 2^{t} - 1\}$. Then
	\begin{align*}
		\epsilon_{j}^{(t)} &= 
		\begin{cases}
			t, & \text{if } j = 1, \\
			t - k, & \text{if } 2^k + 1 \leq j \leq 2^{k+1} - 1, \text{ where } k \in \{1, 2, 3, \dots, t-1\}.
		\end{cases}
	\end{align*}
	
	\begin{proof}
		For $t = 1$, we have $j = 1$, and hence $\epsilon_{1}^{(1)} = 1$, which follows directly from the definition.
		
		For all $t \geq 2$, since $2^{t-1} < 2^{t} - 1 < 2^{t}$, we can deduce that
		\[
		\epsilon_{1}^{(t)} = \lceil \log_2(2^{t} - 1) \rceil = t.
		\]
		
		For $t = 2$, $j \in \{1, 3\}$, so we have $\epsilon_{1}^{(2)} = 2$ and, by Lemma \ref{L10}, $\epsilon_{3}^{(2)} = 1$.
		
		Similarly, for $t = 3$, $j \in \{1, 3, 5, 7\}$. In this case, $\epsilon_{1}^{(3)} = 3$, and from Lemma \ref{L10}, we get:
		\[
		\epsilon_{3}^{(3)} = \epsilon_{3}^{(2)} + 1 = 2, \quad \text{and} \quad \epsilon_{j}^{(3)} = 1 \text{ for } j \in \{5, 7\}.
		\]
		
		By continuing this reasoning for all values of $t$, we obtain the desired result.
		
	\end{proof}
\end{lemma}

For simplicity, we define $\Gamma_{(t)}$ to be the set of all odd integers in $\{1, 2, 3, \dots, 2^{t} - 1\}$, where $t$ is any fixed positive integer. Then for each $j \in \Gamma_{(t)}$ and $i \in B_{j}^{(t)}$, there is a unique $0 \leq \lambda_{ij} \leq \ell_{j} - 1$ such that
\[
i 2^{\lambda_{ij}} \equiv j \pmod{v}.
\]
Then for any $x \in \mathbb{F}_{2^m}$, we have
\begin{align}\label{E3}
	\operatorname{Tr}\left(\sum_{i=1}^{2^{t} - 1} x^i\right) 
	&= \operatorname{Tr}\left(\sum_{j \in \Gamma_{(t)}} \sum_{i \in B_j^{(t)}} x^i \right) \nonumber \\
	&= \sum_{j \in \Gamma_{(t)}} \sum_{i \in B_j^{(t)}} \operatorname{Tr}(x^i) \nonumber \\
	&= \sum_{j \in \Gamma_{(t)}} \kappa_j^{(t)} \operatorname{Tr}(x^j).
\end{align}
For convenience, we define $A=\{1,2,3,\dots,2^{h-1}-1\}$, where $h=\frac{m-1}{2}$.
\vspace{0.5em}
\begin{lemma}
	For any $i, j \in A$ with $i \neq j$, we have $C_{i + 2^h} \cap C_{j + 2^h} = \emptyset$.
	\begin{proof}
		Note that for any $i \in A$, we have $2^h < i + 2^h < 2^{h+1} - 1$. If $i, j \in A$ and $i$ is odd with $i \neq j$, then according to Lemma \ref{lem:L1}, $i + 2^h$ is the coset leader of $C_{i + 2^h}$, and the coset leader of $C_{j + 2^h}$ cannot be equal to $i + 2^h$. Hence, in this case, $C_{i + 2^h} \cap C_{j + 2^h} = \emptyset$.
		
		If $i, j \in A$ and $i$ is even with $i \neq j$, then there is an odd integer $i_1$ such that $i = 2^s i_1$, where $s \in \{1, 2, \dots, h - 2\}$. In this case, $C_{i + 2^h} = C_{i_1 + 2^{h - s}}$. According to Lemma \ref{lem:L1}, $i_1 + 2^{h - s}$ is the coset leader of $C_{i + 2^h}$. Thus, similarly, we have $C_{i + 2^h} \cap C_{j + 2^h} = \emptyset$.
	\end{proof}
\end{lemma}
\vspace{0.5em}
\begin{lemma}\label{L3}
	For any $i,j \in A$, where $j$ is odd, we have
	\begin{gather*}
		C_{i + 2^h} \cap C_j = \begin{cases}
			C_j, \quad &\text{if } (i, j) \text{ is of the form } (2^s i_1, i_1 + 2^{h - s}) \\
			\emptyset, \quad &\text{otherwise}  
		\end{cases}
	\end{gather*}
	where $i_1$ ranges over the odd integers in $\{1, 2, 3, \dots, 2^{h - 1 - s} - 1\}$ and $s \in \{2, 3, \dots, h - 2\}$.
	\begin{proof}
		If $i, j \in A$ and $i$ is odd, by Lemma \ref{lem:L1}, the coset leaders of $C_{i + 2^h}$ and $C_j$ are $i + 2^h$ and $j$, respectively. Since $j < i + 2^h$, we have $C_{i + 2^h} \cap C_j = \emptyset$.
		
		If $i, j \in A$ and $i$ is even, then $i = 2^s i_1$ for some positive odd integer $i_1$. According to Lemma \ref{lem:L1}, the coset leaders of $C_{i + 2^h}$ and $C_j$ are $i_1 + 2^{h - s}$ and $j$, respectively. Note that $C_j = C_{i + 2^h}$ is possible only if $j = i_1 + 2^{h - s}$. Since $i < 2^{h - 1}$ and $j < 2^{h - 1}$, we have $i_1 < 2^{h - 1 - s}$, and $s \in \{2, 3, \dots, h - 2\}$. Hence, the result follows.
	\end{proof}
\end{lemma}
\vspace{0.5em}
\begin{lemma}\label{L7}
	Let $m \geq 5$ be odd and \( s^\infty \) be the sequence defined in Eq. \((\ref{E2})\). Then the generator polynomial \( \textit{g}_s(x)\) corresponding to the sequence \( s^\infty \) is given by
	\begin{align}
		\textit{g}_s(x) &= \prod_{i \in \Gamma_{\left( \frac{m-3}{2} \right)}} m_{\alpha^{-i-2^{\frac{m-1}{2}}}}(x) 
		\prod_{\substack{j=1 \\ \mathbb{N}_2(j)=0}}^{\frac{m-5}{2}} \left( \prod_{j \in \Gamma_{\left(\frac{m-1}{2}-j \right)}} m_{\alpha^{-i-2^{\frac{m+1}{2}-j}}}(x)\right. \nonumber \\ &\mathrel{\phantom{=}} \left.\kern-\nulldelimiterspace \times\; \prod_{j \in \Gamma_{\left( \frac{m-1}{2}-j \right)} \setminus \Gamma_{\left( \frac{m-3}{2}-j \right)}} m_{\alpha^{-i-2^{\frac{m-1}{2}-j}}}(x)\right) 
		 m_{\alpha^{-3}}(x) m_{\alpha^{-1}}(x) (x-1) \nonumber
	\end{align}
	if \( m \equiv 1 \pmod{4} \); and
	\begin{align}
		\textit{g}_s(x) &= \prod_{i \in \Gamma_{\left( \frac{m-3}{2} \right)}} m_{\alpha^{-i-2^{\frac{m-1}{2}}}}(x) 
		\prod_{\substack{j=1 \\ \mathbb{N}_2(j)=0}}^{\frac{m-5}{2}} \left( \prod_{j \in \Gamma_{\left( \frac{m-1}{2}-j \right)}} m_{\alpha^{-i-2^{\frac{m+1}{2}-j}}}(x)  \right. \nonumber \\ &\mathrel{\phantom{=}} \left.\kern-\nulldelimiterspace \times\; \prod_{j \in \Gamma_{\left( \frac{m-1}{2}-j \right)} \setminus \Gamma_{\left( \frac{m-3}{2}-j \right)}} m_{\alpha^{-i-2^{\frac{m-1}{2}-j}}}(x) \right)
		\times m_{\alpha^{-5}}(x) (x-1) \nonumber
	\end{align}
	if \( m \equiv 3 \pmod{4} \). The linear span \( L_s \) corresponding to the sequence \( s^\infty \) is given by
	\begin{flalign*}
		L_s &= \begin{cases}
			1 + m \left( 2^{\frac{m-3}{2}} + 1 \right), & \text{if } m \equiv 1 \pmod{4}, \\
			1 + m \left( 2^{\frac{m-3}{2}} - 1 \right), & \text{if } m \equiv 3 \pmod{4}.
		\end{cases}
	\end{flalign*}
	where \( \Gamma_{(t)} = \{ 1 \leq j \leq 2^t - 1 : \text{$j$ is odd integer} \} \) for any fixed positive integer \( t \), and the map \( \mathbb{N}_2(\cdot) \) is defined by
	\[
	\mathbb{N}_2(j) = \begin{cases} 
		0 & \text{if } 2 \mid j, \\
		1 & \text{if } 2 \nmid j.
	\end{cases}
	\]
	
	\begin{proof}
		For $t=h-1$, combining Lemma $\ref{L4}$ and Eq. $(\ref{E3})$, we obtain
		\begin{align}\label{E4}
			\operatorname{Tr}\left(\sum_{i=1}^{2^{h-1}-1}x^{i}\right) &=
			\operatorname{Tr}\left(\sum_{k=1}^{h-2}\sum_{j\in \Gamma_{(h-k)}\backslash\Gamma_{(h-k-1)}}\kappa_{j}^{(h-1)}x^{j}\right)+\operatorname{Tr}\left(\kappa_{1}^{(h-1)}x\right) \nonumber \\
			&=
			\sum_{k=1}^{h-2}\operatorname{Tr}\left(\sum_{j\in \Gamma_{(h-k)}\backslash\Gamma_{(h-k-1)}}\kappa_{j}^{(h-1)}x^{j}\right)+((h-1) \text{ mod }2)\operatorname{Tr}\left(x\right)
		\end{align}
		According to Lemma \ref{L4}, $\epsilon_{j}^{(h-1)}=(h-1)-(h-k-1)=k$ for all $j\in\Gamma_{(h-k)}\backslash\Gamma_{(h-k-1)}$. It is clear from the right-hand side of $(\ref{E4})$ that $\kappa_{j}^{(h-1)}$ will vanish only for these $j$'s in $\Gamma_{(h-k)}\backslash\Gamma_{(h-k-1)}$ for which $k$ is even, where $k\in\{1,2,\ldots,h-2\}$. Note that for every $j\in\Gamma_{(h-k)}\backslash\Gamma_{(h-k-1)}$, $x^{j}$ can be rewritten in the form $x^{i+2^{h-k-1}}$, where $i\in\Gamma_{(h-k-1)}$. Depending on whether $h$ is even or odd, the remaining terms on the right-hand side of $(\ref{E4})$ are as follows:
		\begin{align}\label{E5}
			\operatorname{Tr}\left(\sum_{i=1}^{2^{h-1}-1}x^{i}\right) &=
			\operatorname{Tr}\left(\sum_{i\in\Gamma_{(h-2)}}^{}x^{i+2^{h-2}}+\sum_{i\in\Gamma_{(h-4)}}^{}x^{i+2^{h-4}}+\dots+\sum_{i\in\Gamma_{(2)}}^{}x^{i+2^2}+x\right)
		\end{align} if $h$ is even; and 
			\begin{align}\label{E6}
				\operatorname{Tr}\left(\sum_{i=1}^{2^{h-1}-1}x^{i}\right) &=
				\operatorname{Tr}\left(\sum_{i\in\Gamma_{(h-2)}}^{}x^{i+2^{h-2}}+\sum_{i\in\Gamma_{(h-4)}}^{}x^{i+2^{h-4}}+\dots+\sum_{i\in\Gamma_{(3)}}^{}x^{i+2^3}+x^3\right)
			\end{align} if $h$ is odd.
		
		\vskip 1pt
		Note that 
		\begin{align}\label{E7}
			\operatorname{Tr}\left(\sum_{i=1}^{2^{h-1}-1}x^{i+2^{h}}\right) &=
			\operatorname{Tr}\left(\sum_{i\in \Gamma_{(h-1)}}^{}x^{i+2^{h}}\right)+\operatorname{Tr}\left(\sum_{i\in A\backslash\Gamma_{(h-1)}}^{}x^{i+2^{h}}\right) \nonumber \\
            &=
			\operatorname{Tr}\left(\sum_{i\in \Gamma_{(h-1)}}^{}x^{i+2^{h}}\right)+\operatorname{Tr}\left(\sum_{i=1}^{2^{h-2}-1}x^{i+2^{h-1}}\right) \nonumber \\
			&= \operatorname{Tr}\left(\sum_{i\in \Gamma_{(h-1)}}^{}x^{i+2^{h}}+\sum_{i\in \Gamma_{(h-2)}}^{}x^{i+2^{h-1}}\right)+\operatorname{Tr}\left(\sum_{i=1}^{2^{h-3}-1}x^{i+2^{h-2}}\right)
		\end{align}
		and 		
		\begin{align}\label{E8}
			\operatorname{Tr}\left(\sum_{i=1}^{2^{h-3}-1}x^{i+2^{h-2}}\right)&=\operatorname{Tr}\left(\sum_{s=2}^{h-2}\sum_{i\in\Gamma_{(h-s-1)}}x^{i+2^{h-s}}\right) \nonumber \\
			&=
			\operatorname{Tr}\left(\sum_{i\in \Gamma_{(h-3)}}x^{i+2^{h-2}}+\sum_{i\in\Gamma_{(h-4)}}x^{i+2^{h-3}}+\dots+\sum_{i\in\Gamma_{(1)}}x^{i+2^{2}}\right)
		\end{align}
		When $h$ is even, from Lemma \ref{L3}, it is clear which terms are the same on the right-hand side of Eq. (\ref{E5}) and (\ref{E8}). 
		\vskip 1pt
		By adding Eq. (\ref{E5}) and (\ref{E8}), we have
		\begin{align}\label{E9}			
			\operatorname{Tr}\left(\sum_{i=1}^{2^{h-1}-1}x^{i}\right)+\operatorname{Tr}\left(\sum_{i=1}^{2^{h-3}-1}x^{i+2^{h-2}}\right) &= \operatorname{Tr}\left(\sum_{i\in\Gamma_{(h-2)}\backslash\Gamma_{(h-3)}}x^{i+2^{h-2}}+\sum_{i\in\Gamma_{(h-4)}}x^{i+2^{h-3}}\right. \nonumber \\   &\mathrel{\phantom{=}} \left.\kern-\nulldelimiterspace +\; \sum_{i\in\Gamma_{(h-4)}\backslash\Gamma_{(h-5)}}x^{i+2^{h-4}}+\dots+\sum_{i\in\Gamma_{(2)}\backslash\Gamma_{(1)}}^{}x^{i+2^2}+x\right)	 
		\end{align}
		With the help of Eq. $(\ref{E7})$ and $(\ref{E9})$, we obtain
		\begin{align}\label{E11}
			\operatorname{Tr}\left( f_{4}(x+1)\right) &= \operatorname{Tr}\left(\sum_{i\in\Gamma_{(h-1)}}x^{i+2^{h}}+\sum_{i\in\Gamma_{(h-2)}}x^{i+2^{h-1}}+\sum_{i\in\Gamma_{(h-2)}\backslash\Gamma_{(h-3)}}x^{i+2^{h-2}}+
			\dots\right. \nonumber \\ &\mathrel{\phantom{=}} \left.\kern-\nulldelimiterspace +\;\sum_{i\in\Gamma_{(2)}}x^{i+2^{3}}+\sum_{i\in\Gamma_{(2)}\backslash\Gamma_{(1)}}x^{i+2^{2}}+x^{3}+x \right)+1
		\end{align}
		Similarly, when $h$ is odd, we obtain
		\begin{align}\label{E10}		
			\operatorname{Tr}\left(f_{4}(x+1)\right) &= \operatorname{Tr}\left(\sum_{i\in\Gamma_{(h-1)}}x^{i+2^{h}}+\sum_{i\in\Gamma_{(h-2)}}x^{i+2^{h-1}}+\sum_{i\in\Gamma_{(h-2)}\backslash\Gamma_{(h-3)}}x^{i+2^{h-2}}\right. \nonumber \\ &\mathrel{\phantom{=}} \left.\kern-\nulldelimiterspace +\;\sum_{i\in\Gamma_{(h-4)}}x^{i+2^{h-3}}+ \dots+\sum_{i\in\Gamma_{(3)}\backslash\Gamma_{(2)}}x^{i+2^3}+\sum_{i\in\Gamma_{(1)}}^{}x^{i+2^2}\right)+1	
		\end{align}
		\vskip 1pt
		Note that, for any integer $t\geq 1$, the number of integers in both sets $\Gamma_{(t)}$ and $\Gamma_{(t+1)}\backslash\Gamma_{(t)}$ is equal to $2^{t-1}$.
		\vskip 1pt
		If $h$ is even, by Lemma \ref{lem:L1} and Eq. (\ref{E11}), we have the linear span of $s^{\infty}$ equals
		\begin{flalign*}
			\hskip 20pt	L_{s}&= (2^{h-2}+2^{h-3}+\dots+2+1)\cdot m+ 2\cdot m+1 &\\
			&=1+m(2^{h-1}+1).
		\end{flalign*}
		\vskip 1pt
		If $h$ is odd, by Lemma \ref{lem:L1} and Eq. (\ref{E10}), we have the linear span of $s^{\infty}$ equals
		\begin{flalign*}
			\hskip 20pt	L_{s}&= (2^{h-2}+2^{h-3}+\dots+2+1)\cdot m+1 &\\
			&=1+m(2^{h-1}-1).
		\end{flalign*}
		Therefore, from Lemma \ref{A} and Eq. $(\ref{E11})$, $(\ref{E10})$ we get the result on the generator polynomial corresponding to the sequence $s^{\infty}$. 
	\end{proof}
\end{lemma}
\textcolor{black}{Lemma \ref{L7} explicitly determines the generator polynomial $\textit{g}_{s}(x)$ of the code $\mathcal{C}_{s}$ as the product of some minimal polynomials over $\mathbb{F}_{2}$. The defining set of $\mathcal{C}_{s}$ is defined to be the set $Z=\{i\in\mathbb{Z}_{v}: \textit{g}_{s}(\alpha^{i})=0\}$. There are a few key methods to determine the lower bounds that can be applied to cyclic codes: BCH bound, Hartmann-Tzeng bound, Roos bound, and Van Lint-Wilson bound. Some of these bounds are easy to use, while others are hard to employ. However, which bound would be beneficial to determine better lower bounds that should be checked by analyzing the structure of the defining set of the code $\mathcal{C}_{s}$.
\vskip 0pt
Hartmann-Tzeng bound \cite{SETA20} states that if there is a set $S$ that contains $\delta-1$ consecutive elements of the defining set $Z$ of $\mathcal{C}_{s}$ and $T=\{jb\text{ mod }v:0\leq j\leq s\}$, where $\operatorname{gcd}(b,v)<\delta$. If $S+T\subseteq Z$ for some $b$ and $s$. Then $d(\mathcal{C}_{s})\geq \delta+ s$. In the following theorem, we determine the upper and lower bound on the minimum distance of the code $\mathcal{C}_{s}$.} 
\vspace{0.5em}
\begin{theorem}\label{EF16}
	Let $m\geq 5$ be odd. The code $\mathcal{C}_{s}$ designed by the sequence $s^{\infty}$ defined in Eq. $(\ref{E2})$, has parameters $[2^{m}-1,2^{m}-1-L_{s},d(\mathcal{C}_{s})]$, where the linear span $L_{s}$ and the generator polynomial $\textit{g}_{s}(x)$ of $\mathcal{C}_{s}$ corresponding to the sequence $s^{\infty}$ are given in Lemma $\ref{L7}$, and the minimum Hamming weight $d(\mathcal{C}_{s})$ is as follows:
	\begin{align*}
		 & \begin{cases}
			\max\{8, 2^{(m-5)/2}+2\}\leq d(\mathcal{C}_{s})\leq 1 + m \left( 2^{\frac{m-3}{2}} + 1 \right) &\text{ if }  m\equiv 1 \hspace{1mm}(\textnormal{mod}\hspace{1mm}4)  \\
			2^{(m-5)/2}+2 \leq d(\mathcal{C}_{s}) \leq 1 + m \left( 2^{\frac{m-3}{2}} - 1 \right) &\text{ if } m\equiv 3 \hspace{1mm}(\textnormal{mod}\hspace{1mm}4) 
		\end{cases}
	\end{align*}
	\begin{proof}The dimension of the code $\mathcal{C}_{s}$ follows from Lemma $\ref{L7}$. Since $x-1$ is a divisor of the generator polynomial $\textit{g}_{s}(x)$, the minimum weight $d(\mathcal{C}_{s})$ must be even. Hence, by applying the Singleton bound \cite{SETA21}, we have $d(\mathcal{C}_{s})\leq L_{s}$. Let $S=\{1+2^{\frac{m-1}{2}}\}$ and $T=\{2j:0\leq j\leq 2^{\frac{m-1}{2}-2}-1\}$. Note that $\operatorname{gcd}(2,v)<2$ and the reciprocal of the generator polynomial $\textit{g}_{s}(x)$ in Lemma $\ref{L7}$ has roots $\alpha^{j}$ for all $j\in S+T=\{1+2^{\frac{m-1}{2}},3+2^{\frac{m-1}{2}},\dots,2^{\frac{m-3}{2}}-1+2^{\frac{m-1}{2}}\}$. As we know, the code with generator polynomial $\textit{g}_{s}(x)$ in Lemma \ref{L7} and the code generated by the reciprocal of $\textit{g}_{s}(x)$ have identical weight distribution, the minimum weight $d(\mathcal{C}_{s})\geq 2^{(m-5)/2}+1$ by applying the Hartmann-Tzeng bound. Hence, $d(\mathcal{C}_{s})\geq 2^{(m-5)/2}+2$. In the case of $m\equiv 1 \text{(mod 4)}$, we have $\mathcal{C}_{s}$ as a subcode of the code $\mathcal{A}_{s}$, as defined in Theorem \ref{L2}. Therefore, the desired conclusion on $d(\mathcal{C}_{s})$ follows by combining all the cases.
\end{proof}
\end{theorem}
\vspace{0.5em}
\begin{example}\label{EX1}
	Let $m = 5$ and $\alpha$ be the root of the primitive polynomial $x^5+x^2+1$ over $\mathbb{F}_{2}$. Then $\textit{g}_{s}(x)=x^{16}+x^{12}+x^{11}+x^{10}+x^{9}+x^{4}+x+1$. $\mathcal{C}_{s}$ is a binary $[31,15,8]$ cyclic code and $\mathcal{C}_{s}^{\perp}$ is a $[31,16,7]$ cyclic code. According to the Database \cite{SETA23}, both codes are optimal.
\end{example}
\vspace{0.5em}
\begin{example}
	Let $m=7$ and $\alpha$ be the root of  the primitive polynomial $x^{7}+x^{3}+1$ over $\mathbb{F}_{2}$. Then $\textit{g}_{s}(x)=x^{22}+x^{17}+x^{16}+x^{14}+x^{13}+x^{12}+x^{11}+x^{8}+x^{7}+x^{2}+x+1$. $\mathcal{C}_{s}$ is a binary $[127,105,6]$ cyclic code and $\mathcal{C}_{s}^{\perp}$ is a $[127,22,43]$ cyclic code. 
\end{example}

\subsection{Binary code $\mathcal{C}_{s}$ from the monomial $x^{2^{m-1}-2^{(m-1)/2}+1}$}
Define $f_{5}(x)=x^{2^{2h}-2^{h}+1}$, where $h=\frac{m-1}{2}$. The monomial $f_{5}(x)$ is known as the \textit{Kasami function}, which is APN as $\operatorname{gcd}(m,h)=1$ (\cite{SETA6}). Since $(2^{2h}-2^{h}+1)=\frac{2^{3h}+1}{2^{h}+1}$ and $\operatorname{gcd}(2^{3h}+1,2^{m}-1)=1$, we have $f_{5}(x)$ is also a permutation over $\mathbb{F}_{2^m}$. 
Let $\Gamma_{(t)}=\{1\leq j\leq 2^{t}-1:\text{$j$ is odd}\}$, where $t$ is any fixed positive integer. Then we have
\begin{flalign}
	\hskip 15pt	\operatorname{Tr}\left(f_{5}(x+1)\right) &= \operatorname{Tr}\left((x+1)(x^{2^{h}}+1)^{\sum_{i=0}^{h-1}2^{i}}\right) \nonumber &\\
	&=\operatorname{Tr}\left((x+1)\sum_{i=0}^{2^{h}-1}x^{i\cdot 2^{h}}\right) \nonumber &\\
	&= \operatorname{Tr}\left(\sum_{i=0}^{2^{h}-1}x^{i+2^{h+1}}+\sum_{i=0}^{2^{h}-1}x^{i}\right) \nonumber &\\
	&= 1+\operatorname{Tr}\left( x\right)+\operatorname{Tr}\left(\sum_{i\in\Gamma_{(h)}}x^{i+2^{h+1}}+\sum_{i=1}^{2^{h-1}-1}x^{i+2^{h}}+\sum_{i=1}^{2^{h}-1}x^{i}\right) \nonumber &
\end{flalign}
The sequence $s^{\infty}$ of $(\ref{B})$ designed from the monomial $f_{5}(x)$ is given by 
\begin{flalign}\label{E14}
	\hskip 30pt	s_{t}&=1+\operatorname{Tr}\left( \alpha^{t}\right)+\operatorname{Tr}\left(\sum_{i\in\Gamma_{(h)}}(\alpha^{t})^{i+2^{h+1}}+\sum_{i=1}^{2^{h-1}-1}(\alpha^{t})^{i+2^{h}}+\sum_{i=1}^{2^{h}-1}(\alpha^{t})^{i}\right), \textnormal{ for all $t\geq 0$.}&
\end{flalign}
This subsection studies the code $\mathcal{C}_{s}$ designed by the sequence $s^{\infty}$ of (\ref{E14}). First, we need to prove some important Lemmas.
\vspace{0.5em}
\begin{lemma}\label{L12}
	For any $j\in\Gamma_{(h)}$, we have
	\begin{enumerate}[(i)]
		\item $j+2^{h+1}$ is the coset leader of $C_{j+2^{h+1}}$ for $j\neq 1$, and the coset leader of $C_{1+2^{h+1}}$ is $1+2^{h}$.
		\item $\ell_{j+2^{h+1}}=|C_{j+2^{h+1}}|=m$.
	\end{enumerate}
	\begin{proof}
		We will start with the first statement. Let $j\in\Gamma_{(h)}$ with $\operatorname{w}_{2}(j)=k\geq 2$, then $j=1+2^{j_{1}}+2^{j_{2}}+\dots+2^{j_{k-1}}$, where $1\leq j_{1}<j_{2}<\dots<j_{k-1}\leq h-1$. According to Lemma \ref{L5}, the coset leader of $C_{j+2^{h+1}}$ must be odd, which means that the coset leader of $C_{j+2^{h+1}}$ must be one of $(j+2^{h+1})2^{m-j_{t}}\text{ (mod $v$)}$ for some $t\in\{1,2,\dots,k-1\}$ or $1+j2^{h}$ or $j+2^{h+1}$ itself. However, since $h+2\leq m-j_{t}\leq m-1$ for each $t\in\{1,2,\dots,k-1\}$, it is not difficult to check that $j+2^{h+1}$ is the smallest odd number in $C_{j+2^{h+1}}$, hence the coset leader. Similarly, $1+2^{h}$ is the coset leader of $C_{1+2^{h+1}}$. \\
		Now we show the second statement. Note that for any $j\in\Gamma_{(h)}$, $(j+2^{h+1})\cdot2^{\ell}< 2^{m}-1$ for any $0\leq \ell\leq h-1$. That means $|C_{j+2^{h+1}}|\geq h=\frac{m-1}{2}$. Since $\operatorname{gcd}(m,2)=1$ and $\ell_{j+2^{h+1}}$ is divisible by $m$, the size of $C_{j+2^{h+1}}$ must be $m$ for all $j\in\Gamma_{(h)}$.
	\end{proof}
\end{lemma}
\vspace{0.5em}
\begin{lemma}\label{L13}
	For any $i\in A$ and $j\in\Gamma_{(h)}$, where $A$ is as defined in Lemma \textnormal{\ref{L3}}, we have
	\begin{enumerate}[(i)]
		\item $C_{i+2^{h}}\cap C_{j+2^{h+1}}\neq\emptyset$ only if $(i,j)=(1,1)$. Moreover, $C_{j}\cap C_{j+2^{h+1}}=\emptyset$.
		\item \begin{gather*}
			C_{i+2^h}\cap C_{j}=	\begin{cases}
				C_{j},\quad \text{ if }(i,j)\text{ is of the form }(2^{s}i_{1},i_{1}+2^{h-s}) \\
				\emptyset,\quad \text{otherwise}  
			\end{cases}
		\end{gather*}
		when $i_{1}$ ranges over the integers in $\Gamma_{(h-1-s)}$ and $s\in\{1,2,3,\dots,h-2\}$.
	\end{enumerate}
	\begin{proof}
		We commence with the first statement. Note that $i+2^{h}<j+2^{h+1}$ for all $i\in A$ and $j\in\Gamma_{(h)}$. When $j\neq 1$, the coset leaders of $C_{i+2^{h}}$ and $C_{j+2^{h+1}}$ are distinct. Hence, in this case, $C_{i+2^{h}}\cap C_{j+2^{h+1}}=\emptyset$.
		\vskip 1pt
		When $j=1$ and $i$ is even, then $i=2^{s}i_{1}$ for some odd positive integer $i_{1}$ with $s\in\{1,2,3,\dots,h-2\}$. According to Lemma \ref{lem:L1}, \ref{L12}, the coset leaders of $C_{i+2^{h}}$ and $C_{1+2^{h+1}}$ are distinct, respectively, $i_{1}+2^{h-s}$ and $1+2^{h}$. Hence, also in this case, $C_{i+2^{h}}\cap C_{j+2^{h+1}}=\emptyset$.
		\vskip 1pt
		When $j=1$ and $i$ is odd. Note that $C_{i+2^{h}}=C_{1+2^{h+1}}$ would imply $\operatorname{w}_{2}(i+2^{h})=\operatorname{w}_{2}(1+2^{h+1})=2$. But, $\operatorname{w}_{2}(i+2^{h})>2$ for $i\neq 1$, and $C_{1+2^{h}}=C_{1+2^{h+1}}$ is obvious. Therefore, $C_{i+2^{h}}\cap C_{1+2^{h+1}}\neq\emptyset$ only if $i=1$. Similarly, one can prove that $C_{j}\cap C_{j+2^{h+1}}=\emptyset$. Hence, the proof of the first statement follows. \\
		\vspace{0.5 mm}
		The second statement can be accomplished in a similar manner as Lemma \ref{L3}.
	\end{proof}
\end{lemma}
\vspace{0.5em}
\begin{lemma}\label{L16}
Let $m \geq 7$ be an odd integer and $s^{\infty}$ be the sequence defined in Eq.~\textnormal{(\ref{E14})}. Then the generator polynomial $\textit{g}_s(x)$ corresponding to the sequence $s^{\infty}$ is given by

\begin{flalign*}
	\textit{g}_s(x) &= \prod_{i \in \Gamma_{(\frac{m-1}{2})} \setminus \{1\}} m_{\alpha^{-i - 2^{\frac{m+1}{2}}}}(x)
	\prod_{i \in \Gamma_{(\frac{m-3}{2})} \setminus \{1\}} m_{\alpha^{-i - 2^{\frac{m-1}{2}}}}(x) \prod_{\substack{j=1\\ \mathbb{N}_2(j)=1}}^{\frac{m-7}{2}} 
	\left( \prod_{i \in \Gamma_{(\frac{m-1}{2} - j)} \setminus \Gamma_{(\frac{m-3}{2} - j)}} 
	m_{\alpha^{-i - 2^{\frac{m-1}{2} - j}}}(x) \right. \nonumber \\ &\mathrel{\phantom{=}} \left.\kern-\nulldelimiterspace \times\; \prod_{i \in \Gamma_{(\frac{m-5}{2} - j)}} m_{\alpha^{-i - 2^{\frac{m-3}{2} - j}}}(x) \right) 
	m_{\alpha^{-3}}(x) m_{\alpha^{-1}}(x) (x-1), &
\end{flalign*}

if $m \equiv 1 \pmod{4}$; and

\begin{flalign*}
	\textit{g}_s(x) &= \prod_{i \in \Gamma_{(\frac{m-1}{2})} \setminus \{1\}} m_{\alpha^{-i - 2^{\frac{m+1}{2}}}}(x)
	\prod_{i \in \Gamma_{(\frac{m-3}{2})} \setminus \{1\}} m_{\alpha^{-i - 2^{\frac{m-1}{2}}}}(x) \prod_{\substack{j=1\\ \mathbb{N}_2(j)=1}}^{\frac{m-7}{2}} 
	\left( \prod_{i \in \Gamma_{(\frac{m-1}{2} - j)} \setminus \Gamma_{(\frac{m-3}{2} - j)}} 
	m_{\alpha^{-i - 2^{\frac{m-1}{2} - j}}}(x)  \right. \nonumber \\ &\mathrel{\phantom{=}} \left.\kern-\nulldelimiterspace \times\; \prod_{i \in \Gamma_{(\frac{m-5}{2} - j)}} m_{\alpha^{-i - 2^{\frac{m-3}{2} - j}}}(x) \right) 
	m_{\alpha^{-7}}(x) (x-1), &
\end{flalign*}

if $m \equiv 3 \pmod{4}$. The linear span $L_s$ corresponding to the sequence $s^{\infty}$ is given by
\begin{flalign*}
	L_s &= 
	\begin{cases}
		1 + m(2^{(m-1)/2} - 1), & \text{if } m \equiv 1 \pmod{4}, \\
		1 + m(2^{(m-1)/2} - 3), & \text{if } m \equiv 3 \pmod{4}.
	\end{cases} &
\end{flalign*}	
where $\Gamma_{(t)} = \{1 \leq j \leq 2^t - 1 : \text{$j$ is an odd integer}\}$ for any fixed positive integer $t$, and the map $\mathbb{N}_2(\cdot)$ is defined by \[
\mathbb{N}_2(j) = \begin{cases} 
	0 & \text{if } 2 \mid j, \\
	1 & \text{if } 2 \nmid j.
\end{cases}
\]
\begin{proof}
	For $t=h$, combining Lemma \ref{L4} and Eq. (\ref{E3}), we obtain
	\begin{align}\label{E16}
		\operatorname{Tr}\left(\sum_{i=1}^{2^{h}-1}x^{i}\right)&= \operatorname{Tr}\left(\sum_{k=1}^{h-1}\sum_{j\in\Gamma_{(h-k+1)}\setminus\Gamma_{(h-k)}}\kappa_{j}^{(h)}x^{j}\right)+\operatorname{Tr}\left(\kappa_{1}^{(h)}x\right) \nonumber \\
		&= \sum_{k=1}^{h-1}\operatorname{Tr}\left(\sum_{j\in\Gamma_{(h-k+1)}\setminus\Gamma_{(h-k)}}\kappa_{j}^{(h)}x^{j}\right)+\left(h \text{ mod }2\right)\operatorname{Tr}\left(x\right)
	\end{align}
	According to Lemma \ref{L4}, $\epsilon_{j}^{(h)}=h-(h-k)=k$ for all $j\in\Gamma_{(h-k+1)}\setminus\Gamma_{(h-k)}$. It is clear from the right-hand side of (\ref{E16}) that $\kappa_{j}^{(h)}$ will vanish only for these $j$'s in $\Gamma_{(h-k+1)}\setminus\Gamma_{(h-k)}$ for which $k$ is even, where $k\in\{1,2,\dots, h-1\}$. Note that for every $j\in\Gamma_{(h-k+1)}\setminus\Gamma_{(h-k)}$, $x^{j}$ can be rewritten in the form $x^{i+2^{h-k}}$, where $i\in\Gamma_{(h-k)}$. Depending on whether $h$ is even or odd, the remaining terms on the right-hand side of (\ref{E16}) are as follows:
	\begin{equation}\label{E17}
		\operatorname{Tr}\left(\sum_{i=1}^{2^{h}-1}x^{i}\right)=\operatorname{Tr}\left(\sum_{i\in\Gamma_{(h-1)}}x^{i+2^{h-1}}+\sum_{i\in\Gamma_{(h-3)}}x^{i+2^{h-3}}+\dots+\sum_{i\in\Gamma_{(3)}}x^{i+2^{3}}+x^{3}\right)
	\end{equation}
	if $h$ is even; and
	\begin{equation}\label{E18}
		\operatorname{Tr}\left(\sum_{i=1}^{2^{h}-1}x^{i}\right)=\operatorname{Tr}\left(\sum_{i\in\Gamma_{(h-1)}}x^{i+2^{h-1}}+\sum_{i\in\Gamma_{(h-3)}}x^{i+2^{h-3}}+\dots+\sum_{i\in\Gamma_{(2)}}x^{i+2^{2}}+x\right)
	\end{equation}
	if $h$ is odd.
	\vskip 1pt
	From Eq. (\ref{E7}) and (\ref{E8}), it is easy to note that
	\begin{equation}\label{E15}
		\operatorname{Tr}\left(\sum_{i=1}^{2^{h-1}-1}x^{i+2^{h}}\right)=\operatorname{Tr}\left(\sum_{i\in\Gamma_{(h-1)}}x^{i+2^{h}}+\sum_{i\in\Gamma_{(h-2)}}x^{i+2^{h-1}}+\dots+\sum_{i\in\Gamma_{(2)}}x^{i+2^{3}}+\sum_{i\in\Gamma_{(1)}}x^{i+2^{2}}\right)
	\end{equation}
	When $h$ is even, by adding (\ref{E17}) and (\ref{E15}), we have
	\begin{align}\label{E20}
		\operatorname{Tr}\left(\sum_{i=1}^{2^{h-1}-1}x^{i+2^{h}}\right)+\operatorname{Tr}\left(\sum_{i=1}^{2^{h}-1}x^{i}\right)&=\operatorname{Tr}\left(\sum_{i\in\Gamma_{(h-1)}}x^{i+2^{h}}+\sum_{i\in\Gamma_{(h-1)}\setminus\Gamma_{(h-2)}}x^{i+2^{h-1}}+\dots\right. \nonumber \\ &\mathrel{\phantom{=}} \left.\kern-\nulldelimiterspace +\;\sum_{i\in\Gamma_{(3)}\setminus\Gamma_{(2)}}x^{i+2^{3}}+\sum_{i\in\Gamma_{(1)}}x^{i+2^{2}}+x^{3}\right)
	\end{align}
	By Lemma \ref{L12}, we conclude that $\operatorname{Tr}\left(x^{1+2^{h+1}}\right)=\operatorname{Tr}\left(x^{1+2^{h}}\right)$. Therefore, 
	\begin{align}\label{E19}
		\operatorname{Tr}\left(f_{5}(x+1)\right)&=\operatorname{Tr}\left(\sum_{i\in\Gamma_{(h)}\setminus\{1\}}x^{i+2^{h+1}}+\sum_{i\in\Gamma_{(h-1)}\setminus\{1\}}x^{i+2^{h}}+\sum_{i\in\Gamma_{(h-1)}\setminus\Gamma_{(h-2)}}x^{i+2^{h-1}}+\dots\right. \nonumber \\ &\mathrel{\phantom{=}} \left.\kern-\nulldelimiterspace +\;\sum_{i\in\Gamma_{(3)}\setminus\Gamma_{(2)}}x^{i+2^{3}}+\sum_{i\in\Gamma_{(1)}}x^{i+2^{2}}+x^{3}+x\right)+1
	\end{align}
	Similarly, when $h$ is odd, we obtain
	\begin{align}\label{E20}
		\operatorname{Tr}\left(f_{5}(x+1)\right)&=\operatorname{Tr}\left(\sum_{i\in\Gamma_{(h)}\setminus\{1\}}x^{i+2^{h+1}}+\sum_{i\in\Gamma_{(h-1)}\setminus\{1\}}x^{i+2^{h}}+\sum_{i\in\Gamma_{(h-1)}\setminus\Gamma_{(h-2)}}x^{i+2^{h-1}}+\dots\right. \nonumber \\ &\mathrel{\phantom{=}} \left.\kern-\nulldelimiterspace +\;\sum_{i\in\Gamma_{(4)}\setminus\Gamma_{(3)}}x^{i+2^{4}}+\sum_{i\in\Gamma_{(2)}}x^{i+2^{3}}+\sum_{i\in\Gamma_{(2)}\setminus\Gamma_{(1)}}x^{i+2^{2}}\right)+1
	\end{align}
    From Lemma \ref{lem:L1} and \ref{L13}, it is evident that none of the terms on the right-hand side of Eq. (\ref{E19}) and (\ref{E20}) will mutually cancel out.  
    \vskip 1pt
		Note that, for any integer $t\geq 1$, the number of integers in both sets $\Gamma_{(t)}$ and $\Gamma_{(t+1)}\backslash\Gamma_{(t)}$ is equal to $2^{t-1}$.
        \vskip 1pt
        When $h$ is even, by Lemma \ref{lem:L1}, \ref{L12} and Eq. (\ref{E19}), we have the linear span of $s^{\infty}$ as follows:
        \begin{flalign*}
			\hskip 20pt	L_{s}&= \{(2^{h-1}-1)+(2^{h-2}-1)+2^{h-3}+\dots+2+1)\}\cdot m+ 2\cdot m+1 &\\
			&=1+m(2^{h}-1).
		\end{flalign*}
        When $h$ is odd, by Lemma \ref{lem:L1}, \ref{L12} and Eq. (\ref{E20}), we have the linear span of $s^{\infty}$ as follows:
        \begin{flalign*}
			\hskip 20pt	L_{s}&= \{(2^{h-1}-1)+(2^{h-2}-1)+2^{h-3}+\dots+2+1)\}\cdot m+1 &\\
			&=1+m(2^{h}-3).
		\end{flalign*}
      Therefore, from Lemma \ref{C} and Eq. (\ref{E19}), (\ref{E20}) we get the result on the generator polynomial corresponding to the sequence $s^{\infty}$.
\end{proof}
\end{lemma}
\vspace{0.5em}
\begin{theorem}\label{T2}
	Let $m\geq 7$ be odd. The code $\mathcal{C}_{s}$ designed by the sequence $s^{\infty}$ defined in Eq. \textnormal{(\ref{E14})}, has parameters $[2^{m}-1,2^{m}-1-L_{s}, d(\mathcal{C}_{s})]$, where the linear span $L_{s}$ and the generator polynomial $\textit{g}_{s}(x)$ of $\mathcal{C}_{s}$ corresponding to the sequence $s^{\infty}$ are given in Lemma \textnormal{\ref{L16}}, and the minimum Hamming weight $2^{\frac{m-3}{2}}\leq d(\mathcal{C}_{s})\leq L_{s}$.
	\begin{proof}
		The dimension of the code $\mathcal{C}_{s}$ follows from Lemma \ref{L16}. As $\mathcal{C}_{s}$ is an even-weight code, we have the minimum weight $d(\mathcal{C}_{s})\leq L_{s}$ by applying the Singleton bound. Let $S=\{3+2^{h+1}\}$ and $T=\{2j:0\leq j\leq 2^{h-1}-2\}$, it is not difficult to verify that the reciprocal of $\textit{g}_{s}(x)$ given in Lemma \textnormal{\ref{L16}} has the roots $\alpha^{j}$ for all $j$ in $S+T$. Since $\operatorname{gcd}(2,2^m-1)<2$, by applying the Hartmann-Tzeng bound, we obtain the minimum weight $d(\mathcal{C}_{s})\geq 2^{h-1}$. Hence, the desired result follows.
	\end{proof}
\end{theorem}
\vspace{0.5em}
\begin{example}
	Let $m=7$ and $\alpha$ be the root of the primitive polynomial $x^7+x^3+1$ over $\mathbb{F}_{2}$. Then
	$\textit{g}_{s}(x)=x^{36}+x^{34}+x^{32}+x^{31}+x^{29}+x^{28}+x^{26}+x^{22}+x^{20}+x^{18}+x^{17}+x^{15}+x^{12}+x^{11}+x^{10}+x^9+x^8+x^7+x^6+x^5+x^4+x^3+x^2+1$. Hence, $\mathcal{C}_{s}$ is a binary $[127,91,8]$ cyclic code.
\end{example}
\vspace{0.5em}
\begin{remark}
	\textcolor{black}{It is interesting to mention that both the trinomials $f_{1}(x)$ in Section $\ref{3.1}$ and $f_{3}(x)$ in Section $\ref{3.3}$ produce an optimal family of cyclic codes with parameters $[2^m-1,2^m-2-m,4]$, where $m\geq 3$ is odd and $[2^m-1,2^m-2-3m,8]$, where $m\geq 5$ is odd, respectively. When $m=5$, the code $\mathcal{C}_{s}$ from the trinomial $f_{5}(x)$ over $\mathbb{F}_{2^m}$ is the same as in Example $7$ of $\cite{SETA3}$. Although for $m=5$, the codes $\mathcal{C}_{s}$ and $\mathcal{C}_{s}^{\perp}$ constructed from the trinomials $f_{2}(x)$ in Example $\ref{EX2}$, $f_{4}(x)$ in Example $\ref{EX1}$, and $f_{5}(x)$ give optimal parameters, but for larger values of $m$ the codes $\mathcal{C}_{s}$ and $\mathcal{C}_{s}^{\perp}$ do not guarantee optimality. Therefore, we must carefully choose a permutation polynomial (or a polynomial with low differential uniformity) that could produce cyclic codes with optimal parameters or cyclic codes with dimensions larger than half its length and the minimum distance close to the square root bound.} 
\end{remark}

\section{Binary cyclic codes from polynomials over $\mathbb{F}_{2^m}$, $m$ is even}
This section focuses on the cyclic codes $\mathcal{C}_{s}$ designed by the sequence $s^{\infty}$ defined in Eq. (\ref{B}) from three classes of permutation trinomials of the form
\begin{align}\label{E21}
	F(x)=x+x^{s(2^{m/2}-1)+1}+x^{t(2^{m/2}-1)+1}
\end{align}
where $m$ is even and $1\leq s,t\leq 2^{m/2}$. 
\vskip 1pt
According to Theorem 3.4 in \cite{SETA2}, for the case when $t=-s$ the polynomial $F(x)$ in Eq. (\ref{E21}) is a permutation over $\mathbb{F}_{2^m}$ if and only if either $m\equiv 0\hspace{1 mm}(\text{mod 4})$ or $m\equiv 2\hspace{1 mm}(\text{mod 4})$ and $\operatorname{exp}_{3}(l)\geq\operatorname{exp}_{3}(2^{m/2}+1)$, where $\operatorname{exp}_{3}(i)$ denotes the exponent of $3$ in the canonical factorization of $i$. Since $(2^{m/2}+1)\equiv 0\hspace{1 mm}(\text{mod 3})$ for any integer $m$ satisfying $m\equiv 2\hspace{1 mm}(\text{mod 4})$, $F(x)$ is not a permutation over $\mathbb{F}_{2^m}$ when $m\equiv 2\hspace{1 mm}(\text{mod 4})$ and $(s,t)\in\{(1,-1),(2,-2)\}$.
\vskip 0pt
According to Theorem 4.7 in \cite{SETA15}, for the case when $(s,t)=(1,2^{\frac{m}{2}-1})$ the trinomial $F(x)$ in Eq. (\ref{E21}) is a permutation over $\mathbb{F}_{2^m}$ if and only if $m\not\equiv 0\hspace{1 mm}(\text{mod 6})$.
\vskip 0pt
Throughout this section, we define $h=\frac{m}{2}$ and use the notation $\Gamma_{(t)}$ as defined before.

\subsection{Binary code $\mathcal{C}_{s}$ from the trinomial of the form (\ref{E21}), where $(s,t)=(1,-1)$}
Define $f_{6}(x)=x+x^{2^{m/2}}+x^{2^m-2^{m/2}+1}$, where $m$ is an even integer. Then we have
\begin{flalign}\label{E13}
	\hskip 30pt	s_{t}&=\operatorname{Tr}\left(f_{6}(\alpha^{t}+1)\right)\nonumber\\ &= \operatorname{Tr}(\alpha^{t}+1)+\operatorname{Tr}\left((\alpha^{t}+1)^{2^{h}}\right)+\operatorname{Tr}\left((\alpha^{t}+1)^{2^{2h}-2^{h}+1}\right) \nonumber \\
	&= \operatorname{Tr}\left((\alpha^{t}+1)(\alpha^{t}+1)^{\sum_{i=0}^{h-1}2^{h+i}}\right)\nonumber \\
	&=\operatorname{Tr}\left((\alpha^{t}+1)\sum_{i=0}^{2^{h}-1}(\alpha^{t})^{i\cdot 2^{h}}\right)\nonumber \\
	&=\operatorname{Tr}\left(\sum_{i=0}^{2^{h}-1}(\alpha^{t})^{i+2^{h}}\right)+\operatorname{Tr}\left(\sum_{i=0}^{2^{h}-1}(\alpha^{t})^{i}\right) \nonumber \\
	&=\operatorname{Tr}\left(\sum_{i=0}^{2^{h+1}-1}(\alpha^{t})^{i}\right) &
\end{flalign}
\vspace{0.5em}
\begin{theorem}\label{L17}
	Let $m\geq 4$ be even and $s^{\infty}$ be the sequence defined in Eq. $(\ref{E13})$. Then the generator polynomial $\textit{g}_{s}(x)$ corresponding to the sequence $s^{\infty}$ is given by
	\begin{flalign*}
		\hskip 30pt		\textit{g}_{s}(x) &= \prod_{i\in\Gamma_{(\frac{m}{2})}\setminus\{1\}}m_{\alpha^{-i-2^{\frac{m}{2}}}}(x)\prod_{\substack{j=1\\ \mathbb{N}_{2}(j)=0}}^{\frac{m-4}{2}}\left(\prod_{i\in\Gamma_{(j)}}m_{\alpha^{-i-2^{j}}}(x)\right)m_{\alpha^{-1}}(x), &
	\end{flalign*}if $m\equiv 0\hspace{1mm}(mod\hspace{1mm}4)$ and
	\begin{flalign*}
		\hskip 30pt		\textit{g}_{s}(x) &= \prod_{i\in\Gamma_{(\frac{m}{2})}\setminus\{1\}}m_{\alpha^{-i-2^{\frac{m}{2}}}}(x)\prod_{\substack{j=1\\ \mathbb{N}_{2}(j)=1}}^{\frac{m-4}{2}}\left(\prod_{i\in\Gamma_{(j)}}m_{\alpha^{-i-2^{j}}}(x)\right), &
	\end{flalign*}if $m\equiv 2\hspace{1mm}(mod\hspace{1mm}4)$; where the map $\mathbb{N}_{2}(\cdot)$ is defined by	\[
	\mathbb{N}_2(j) = \begin{cases} 
	0 & \text{if } 2 \mid j, \\
	1 & \text{if } 2 \nmid j.
	\end{cases}
	\]   
	The linear span $L_{s}$ corresponding to the sequence $s^{\infty}$ is given by
	\begin{flalign*}
		\hskip 30pt	L_{s}&= \begin{cases}
			m\left(\frac{2^{\frac{m}{2}+1}-2}{3}\right),\text{ if }m\equiv 0 \hspace{1mm}(mod\hspace{1mm}4)   \\
			m\left(\frac{2^{\frac{m}{2}+1}-4}{3}\right),\text{ if } m\equiv 2\hspace{1mm}(mod\hspace{1mm}4).
		\end{cases} &
	\end{flalign*}
	Moreover, the code $\mathcal{C}_{s}$ has parameters $[2^{m}-1, 2^{m}-1-L_{s}, d(\mathcal{C}_{s})]$, where $2^{\frac{m-2}{2}}\leq d(\mathcal{C}_{s})\leq L_{s}+1$.
	\begin{proof}
		Note that $\operatorname{Tr}(1)=0$ as $m$ is even and $\ell_{2^{m/2}+1}=|C_{2^{m/2}+1}|=m/2$. By using the properties of the trace function we have $\operatorname{Tr}(x^{2^{m/2}+1})=0$ for all $x\in\mathbb{F}_{2^{m}}$. For $t=h+1$, proceeding similarly to Lemma \ref{L7}, we obtain
		\begin{align}\label{E24}
			\operatorname{Tr}\left(\sum_{i=0}^{2^{h+1}-1}x^{i}\right) &=\operatorname{Tr}\left(\sum_{i\in\Gamma_{(h)}\setminus\{1\}}x^{i+2^{h}}+\sum_{i\in\Gamma_{(h-2)}}x^{i+2^{h-2}}+\dots+\sum_{i\in\Gamma_{(2)}}x^{i+2^{2}}+x\right)
		\end{align}
		if $h$ is even; and
		\begin{align}\label{E25}
			\operatorname{Tr}\left(\sum_{i=0}^{2^{h+1}-1}x^{i}\right) &=\operatorname{Tr}\left(\sum_{i\in\Gamma_{(h)}\setminus\{1\}}x^{i+2^{h}}+\sum_{i\in\Gamma_{(h-2)}}x^{i+2^{h-2}}+\dots+\sum_{i\in\Gamma_{(3)}}x^{i+2^{3}}+x^{3}\right)
		\end{align}
		if $h$ is odd.
		\vskip 1pt
		Note that, for any integer $t\geq 1$, the number of integers in the set $\Gamma_{(t)}$ is equal to $2^{t-1}$. If $h$ is even, by Lemma \ref{lem:L1} and Eq. (\ref{E24}), we have the linear span of $s^{\infty}$ equals
		\begin{flalign*}
			\hskip 20pt	L_{s}&= \left((2^{h-1}-1)+2^{h-3}+\dots+2+1\right)\cdot m & \\
			&=\frac{m(2^{h+1}-2)}{3}.
		\end{flalign*}
		If $h$ is odd, by Lemma \ref{lem:L1} and Eq. (\ref{E25}), we have the linear span of $s^{\infty}$ equals 
		\begin{flalign*}
			\hskip 20pt	L_{s}&= \left((2^{h-1}-1)+2^{h-3}+\dots+2^2+1\right)\cdot m & \\
			&=\frac{m(2^{h+1}-4)}{3}.
		\end{flalign*}
		From Lemma \ref{A} and Eq. (\ref{E24}) and (\ref{E25}) we get the result on the generator polynomial corresponding to the sequence $s^{\infty}$. 
        \vskip 1pt
        The upper bound of the minimum weight $d(\mathcal{C}_{s})$ follows from the Singleton bound. Let $S=\{3+2^{h}\}$ and $T=\{2j:0\leq j\leq 2^{h-1}-2\}$, it can be easily checked that the reciprocal of the generator polynomial $\textit{g}_{s}(x)$ has the roots $\alpha^{j}$ for all $j\in S+T$. Since $\operatorname{gcd}(2,2^{m}-1)<2$, by applying the Hartmann-Tzeng bound, we have the minimum Hamming weight $d(\mathcal{C}_{s})\geq 2^{h-1}$.
	\end{proof}
\end{theorem}	
\vspace{0.5em}
\begin{example}
	Let $m=4$ and $\alpha$ be the root of the primitive polynomial $x^4+x+1$ over $\mathbb{F}_{2}$. The generator polynomial of $\mathcal{C}_{s}$ is $\textit{g}_{s}(x)=x^8+x^7+x^5+x^4+x^3+x+1$. Then $\mathcal{C}_{s}$ is a $[15,7,3]$ cyclic code and $\mathcal{C}_{s}^{\perp}$ is an optimal $[15,8,4]$ cyclic code. The optimal binary $[15,8,4]$ linear code is not cyclic in the Database $\cite{SETA23}$.
\end{example}
\vspace{0.5em}
\begin{example}
    Let $m=6$ and $\alpha$ be the root of the primitive polynomial $x^6+x+1$ over $\mathbb{F}_{2}$. The generator polynomial of $\mathcal{C}_{s}$ is $\textit{g}_{s}(x)=x^{24}+x^{23}+x^{20}+x^{16}+x^{13}+x^{12}+x^{11}+x^8+x^4+x+1$. Then $\mathcal{C}_{s}$ is a $[63,39,7]$ binary cyclic code and its dual $\mathcal{C}_{s}^{\perp}$ is a $[63,24,12]$ cyclic code.
\end{example}
\vspace{0.5em}
\begin{example}
    Let $m=8$ and $\alpha$ be the root of the primitive polynomial $x^8 + x^4 + x^3 + x + 1$ over $\mathbb{F}_{2}$. The generator polynomial of $\mathcal{C}_{s}$ is $\textit{g}_{s}(x)=x^{80} + x^{79} +
	x^{78} + x^{77} + x^{76} + x^{75} + x^{72} + x^{71} + x^{65} + x^{63} + x^{62} + x^{59} + x^{57} +
	x^{56} + x^{53} + x^{49} + x^{48} + x^{46} + x^{45} + x^{44} + x^{43} + x^{40} + x^{34} + x^{33} +
	x^{32} + x^{31} + x^{30} + x^{29} + x^{27} + x^{22} + x^{21} + x^{18} + x^{15} + x^{13} + x^{10} +
	x^7 + x^6 + x^4 + x^2 + x + 1$. Then, by using a Magma program, we have $\mathcal{C}_{s}$ is a $[255,175,d(\mathcal{C}_{s})]$ binary cyclic code, where $15\leq d(\mathcal{C}_{s})\leq 17$ and its dual $\mathcal{C}_{s}^{\perp}$ is a $[255,80,40]$ cyclic code.
\end{example}
\subsection{ Binary code $\mathcal{C}_{s}$ from the trinomial of the form (\ref{E21}), where $(s,t)=(2,-2)$}
Define $f_{7}(x)=x+x^{2^{m/2+1}-1}+x^{2^{m}-2^{m/2+1}+2}$, where $m$ is an even integer. As $m$ being even, $\operatorname{Tr}\left(1\right)=0$. Then we have
\begin{flalign}\label{E23}
	\hskip 30pt	s_{t}&=\operatorname{Tr}\left(f_{7}(\alpha^{t}+1)\right)\nonumber\\
	&=\operatorname{Tr}\left(\alpha^{t}+1\right)+\operatorname{Tr}\left((\alpha^{t}+1)^{2^{h+1}-1}\right)+\operatorname{Tr}\left(((\alpha^{t})^{2}+1)((\alpha^{t})^{2^{h+1}}+1)^{2^{h-1}-1}\right) \nonumber\\
	&= \operatorname{Tr}\left(\alpha^{t}+1\right)+\operatorname{Tr}\left(\sum_{i=0}^{2^{h+1}-1}(\alpha^{t})^{i}\right)+\operatorname{Tr}\left(((\alpha^{t})^{2}+1)\sum_{i=0}^{2^{h-1}-1}(\alpha^{t})^{i\cdot2^{h+1}}\right) \nonumber\\
	&=\operatorname{Tr}\left(\sum_{i\in\Gamma_{(h)}}(\alpha^{t})^{i+2^{h}}+\sum_{i=0}^{2^{h-1}-1}(\alpha^{t})^{i}\right)+\operatorname{Tr}\left(\sum_{i=1}^{2^{h-1}-1}(\alpha^{t})^{i+2^{h}}+\sum_{i=0}^{2^{h-1}-1}(\alpha^{t})^{i}\right) \nonumber\\
	&=\operatorname{Tr}\left(\sum_{i\in\Gamma_{(h)}}(\alpha^{t})^{i+2^{h}}+\sum_{i=1}^{2^{h-1}-1}(\alpha^{t})^{i+2^{h}}\right)&
\end{flalign}
For convenience, we define $\bar{A}=\{1,2,3,\dots,2^{h-1}-1\}$, where $h=\frac{m}{2}$.
\vspace{0.5em}
\begin{lemma}\label{L15}
	For any $i\in \bar{A}$ and $j\in\Gamma_{(h)}$, we have $C_{i+2^{h}}\cap C_{j+2^{h}}\neq \emptyset$ only if $i=j$ with $j\in\Gamma_{(h-1)}$.
	\begin{proof}
		Note that when $i\in \bar{A}$ and $j\in\Gamma_{(h)}\setminus\Gamma_{(h-1)}$, we have $i+2^{h}<j+2^{h}$. According to Lemma \ref{lem:L1}, $j+2^{h}$ is the coset leader of $C_{j+2^{h}}$. This implies that the coset leaders of $C_{i+2^{h}}$ and $C_{j+2^{h}}$ are distinct. Hence, in this case, $C_{i+2^{h}}\cap C_{j+2^{h}}= \emptyset$. \\
		When $i\in \bar{A}$ and $j\in\Gamma_{(h-1)}$, the coset leaders of $C_{i+2^{h}}$ and $C_{j+2^{h}}$ are equal, only if $i=j$. Hence, the result follows.
	\end{proof}
\end{lemma}
\vspace{0.5em}
\begin{theorem}\label{L18}
	Let $m\geq 6$ be even and $s^{\infty}$ be the sequence defined in Eq. $(\ref{E23})$. Then the generator polynomial $\textit{g}_{s}(x)$ corresponding to the sequence $s^{\infty}$ is given by
	\begin{flalign*}
		\hskip 30pt		\textit{g}_{s}(x)&=\prod_{i\in\Gamma_{(\frac{m}{2})}\setminus\Gamma_{(\frac{m-2}{2})}}m_{\alpha^{-i-2^{\frac{m}{2}}}}(x)\prod_{j=1}^{\frac{m-4}{2}}\left(\prod_{i\in\Gamma_{(j)}}m_{\alpha^{-i-2^{j+1}}}(x)\right)&
	\end{flalign*} and the linear span corresponding to the sequence $s^{\infty}$ is equal to $m\cdot(2^{(m-2)/2}-1)$ and 
	Moreover, the code $\mathcal{C}_{s}$ has parameters $[2^{m}-1,2^{m}-1-m(2^{(m-2)/2}-1), d(\mathcal{C}_{s})]$, where $2^{\frac{m-4}{2}}+1\leq d(\mathcal{C}_{s})\leq 1+m(2^{(m-2)/2}-1)$. 
	\begin{proof}
		The proof of this lemma can be easily carried out with the help of Lemma \ref{L15} and \ref{lem:L1}, similar to Lemma \ref{L7}. The upper and lower bound on $d(\mathcal{C}_{s})$ follows from the Singleton bound and the Hartmann-Tzeng bound, respectively.
	\end{proof}
\end{theorem}
\vspace{0.5em}
\begin{example}
	Let $m=4$ and $\alpha$ be the root of the primitive polynomial $x^4+x+1$ over $\mathbb{F}_{2}$. The generator polynomial of $\mathcal{C}_{s}$ is $\textit{g}_{s}(x)=x^4+x+1$. Then $\mathcal{C}_{s}$ is a binary $[15,11,3]$ cyclic code and $\mathcal{C}_{s}^{\perp}$ is a $[15,4,8]$ cyclic code. According to the Database $\cite{SETA23}$, both codes are optimal, and none of the binary linear codes with the same parameters are cyclic. 
\end{example}
\vspace{0.5em}
\begin{example}
    Let $m=6$ and $\alpha$ be the root of the primitive polynomial $x^6+x+1$ over $\mathbb{F}_{2}$. The generator polynomial of $\mathcal{C}_{s}$ is $\textit{g}_{s}(x)=x^{18}+x^{17}+x^{15}+x^{14}+x^{13}+x^{11}+x^{10}+x^8+x^3+x+1$. Then $\mathcal{C}_{s}$ is a $[63,45,5]$ binary cyclic code and its dual $\mathcal{C}_{s}^{\perp}$ is a $[63,18,16]$ cyclic code.
\end{example}
\vspace{0.5em}
\begin{example}
    Let $m=8$ and $\alpha$ be the root of the primitive polynomial $x^8+x^4+x^3+x+1$ over $\mathbb{F}_{2}$. The generator polynomial of $\mathcal{C}_{s}$ is $\textit{g}_{s}(x)=x^{56}+x^{54}+x^{53}+x^{52}+x^{51}+x^{49}+x^{48} + x^{45} + x^{44} + x^{42} + x^{38} + x^{36} + x^{34} +
	x^{33} + x^{31} + x^{30} + x^{29} + x^{27} + x^{25} + x^{24} + x^{20} + x^{19} + x^{17} + x^{16} +
	x^{15} + x^{14} + x^{12} + x^{10} + x^8 + x^7 + x^4 + x^3 + 1$. Then $\mathcal{C}_{s}$ is a $[255,199,10]$ binary cyclic code and its dual $\mathcal{C}_{s}^{\perp}$ is a $[255,56,64]$ cyclic code.
\end{example}
\subsection{Binary code $\mathcal{C}_{s}$ from the trinomial of the form (\ref{E21}), where $(s,t)=(1,2^{\frac{m}{2}-1})$}
Define $f_{8}(x)=x+x^{2^{m/2}}+x^{2^{m-1}-2^{m/2-1}+1}$, where $m$ is an even integer. Note that $\operatorname{Tr}(x^2)=\operatorname{Tr}(x)$ for all $x\in\mathbb{F}_{2^m}$ and $\operatorname{Tr}(1)=0$ for an even integer $m$. Then we have
\begin{flalign}\label{EF8}
	\hskip 30pt	s_{t}&=\operatorname{Tr}\left(f_{8}(\alpha^{t}+1)\right)\nonumber\\ &= \operatorname{Tr}(\alpha^{t}+1)+\operatorname{Tr}\left((\alpha^{t}+1)^{2^{h}}\right)+\operatorname{Tr}\left((\alpha^{t}+1)^{2^{2h-1}-2^{h-1}+1}\right) \nonumber \\
	&= \operatorname{Tr}\left((\alpha^{t}+1)(\alpha^{t}+1)^{\sum_{i=0}^{h-1}2^{h-1+i}}\right)\nonumber \\
    &=\operatorname{Tr}\left((\alpha^{t}+1)\prod_{i=0}^{h-1}((\alpha^{t})^{2^{h-1+i}}+1)\right)\nonumber \\
	&=\operatorname{Tr}\left((\alpha^{t}+1)\sum_{i=0}^{2^{h}-1}(\alpha^{t})^{i\cdot 2^{h-1}}\right)\nonumber \\
	&=\operatorname{Tr}\left(\sum_{i=0}^{2^{h}-1}(\alpha^{t})^{1+i\cdot2^{h-1}}\right)+\operatorname{Tr}\left(\sum_{i=0}^{2^{h}-1}(\alpha^{t})^{i}\right) \nonumber \\
	&=\operatorname{Tr}(\alpha^{t})+\operatorname{Tr}\left(\sum_{i=1}^{2^{h}-1}(\alpha^{t})^{i+2^{h+1}}\right)+\operatorname{Tr}\left(\sum_{i=1}^{2^{h}-1}(\alpha^{t})^{i}\right) &
\end{flalign}
\begin{lemma}\label{EF15}
    For any $j\in\Gamma_{(h)}$, we have
    \begin{enumerate}
        \item[(i)] $j+2^{h+1}$ is the coset leader of $C_{j+2^{h+1}}$ for $j\not\in\Gamma_{(2)}$, and the coset leaders of $C_{1+2^{h+1}}$ and $C_{3+2^{h+1}}$ are $1+2^{h-1}$ and $1+2^{h-1}+2^{h}$ respectively.
        \item[(ii)] $\ell_{j+2^{h+1}}=|C_{j+2^{h+1}}|=m$ except that $\ell_{5+2^{h+1}}=|C_{5+2^{h+1}}|=2$ when $h=3$.
    \end{enumerate}
    \begin{proof} The first statement of this lemma can be easily modified similarly to Lemma \ref{L12}.
    \vskip 1pt
    We proof the second statement of this lemma. By Lemma \ref{lem:L1}, we have $\ell_{1+2^{h-1}}=\ell_{1+2^{h-1}+2^{h}}=m$. Note that for any $j\in\Gamma_{(h)}\setminus\Gamma_{(2)}$,  $(j+2^{h+1})\cdot2^{\ell}<2^{m}-1$ for any $0\leq \ell\leq h-2$. That means $|\ell_{j+2^{h+1}}|\geq h-1$. If possible for some $j\in\Gamma_{(h)}\setminus\Gamma_{(2)}$, $\ell_{j+2^{h+1}}<m$. Then $\ell_{j+2^{h+1}}\leq \frac{m}{2}=h$.
    \vskip 0pt
    Suppose that $2^{h}\cdot(j+2^{h+1})\equiv 2+j\cdot2^{h}\equiv (j+2^{h+1})\hspace{1mm}(\textnormal{mod}\hspace{1mm}2^{m}-1)$, which implies $j\equiv 2\hspace{1mm}(\textnormal{mod}\hspace{1mm}2^{h}+1)$. This is not possible because $j\neq 2$ and $j-2< 2^{h}+1$. Therefore $\ell_{j+2^{h+1}}\neq h$.
    \vskip 0pt 
        On the other hand, since $\ell_{j+2^{h+1}}$ divides $m$ and $\frac{2h}{h-1}$ is an integer only when $h\in\{2,3\}$. One can easily check $\ell_{5+2^{h+1}}=h-1$ for $h=3$. Hence, the proof.
    \end{proof}
\end{lemma}
For convenience, we define $\bar{\bar{A}}=\{1,2,3,\cdots 2^{h}-1\}$, where $h=\frac{m}{2}$.
\vspace{0.5em}
\begin{lemma}\label{EF14}
    For any $i\in\bar{\bar{A}}$ and $j\in\Gamma_{(h)}$, we have
    \item \begin{gather*}
			C_{i+2^{h+1}}\cap C_{j}=	\begin{cases}
				C_{j},\quad \text{ if }(i,j)\text{ is of the form }(2^{s}i_{1},i_{1}+2^{h+1-s}) \\
				\emptyset,\quad \text{otherwise}  
			\end{cases}
		\end{gather*}
        when $i_{1}$ ranges over the integers in $\Gamma_{(h-s)}$ and $s\in\{2,3,\cdots,h-1\}$.
        \begin{proof}
            The proof of this lemma can be easily modified similarly with the help of Lemma \ref{L3}.
        \end{proof}
\end{lemma}
\vspace{0.5em}
\begin{theorem}
    Let $m\geq 6$ be even and $s^{\infty}$ be the sequence defined in Eq. $(\ref{EF8})$. Then the generator polynomial $\textit{g}_{s}(x)$ corresponding to the sequence $s^{\infty}$ is given by 
    \begin{flalign*}
		\hskip 30pt		\textit{g}_{s}(x) &= \prod_{i\in\Gamma_{(\frac{m}{2})}}m_{\alpha^{-i-2^{\frac{m}{2}+1}}}(x)\prod_{i\in\Gamma_{(\frac{m}{2}-1)}}m_{\alpha^{-i-2^{\frac{m}{2}}}}(x)\prod_{\substack{j=1\\ \mathbb{N}_{2}(j)=1}}^{\frac{m-6}{2}}\left(\prod_{i\in\Gamma_{(\frac{m}{2}-j)}\backslash\Gamma_{(\frac{m-2}{2}-j)}}m_{\alpha^{-i-2^{\frac{m}{2}-j}}}(x)\right. \nonumber \\ &\mathrel{\phantom{=}} \left.\kern-\nulldelimiterspace \times\; \prod_{i\in\Gamma_{(\frac{m-4}{2}-j)}}m_{\alpha^{-i-2^{\frac{m-2}{2}-j}}}(x)\right)m_{\alpha^{-3}}(x)m_{\alpha^{-1}}(x), &
	\end{flalign*}
    if $m\equiv 0\hspace{1mm}(mod\hspace{1mm}4)$ and
	\begin{flalign*}
		\hskip 30pt		\textit{g}_{s}(x) &= \prod_{i\in\Gamma_{(\frac{m}{2})}}m_{\alpha^{-i-2^{\frac{m}{2}+1}}}(x)\prod_{i\in\Gamma_{(\frac{m}{2}-1)}}m_{\alpha^{-i-2^{\frac{m}{2}}}}(x)\prod_{\substack{j=1\\ \mathbb{N}_{2}(j)=1}}^{\frac{m-6}{2}}\left(\prod_{i\in\Gamma_{(\frac{m}{2}-j)}\backslash\Gamma_{(\frac{m-2}{2}-j)}}m_{\alpha^{-i-2^{\frac{m}{2}-j}}}(x)\right. \nonumber \\ &\mathrel{\phantom{=}} \left.\kern-\nulldelimiterspace \times\;\prod_{i\in\Gamma_{(\frac{m-4}{2}-j)}}m_{\alpha^{-i-2^{\frac{m-2}{2}-j}}}(x)\right)m_{\alpha^{-7}}(x), &
	\end{flalign*}if $m\equiv 2\hspace{1mm}(mod\hspace{1mm}4)$; where the map $\mathbb{N}_{2}(\cdot)$ is defined by	\[
	\mathbb{N}_2(j) = \begin{cases} 
	0 & \text{if } 2 \mid j, \\
	1 & \text{if } 2 \nmid j.
	\end{cases}
	\]   
	The linear span $L_{s}$ corresponding to the sequence $s^{\infty}$ is given by
	\begin{flalign*}
		\hskip 30pt	L_{s}&= \begin{cases}
			m\left(2^{\frac{m}{2}}+1\right)-\frac{m}{2},\text{ if }m\equiv 0 \hspace{1mm}(mod\hspace{1mm}4);   \\
			m\left(2^{\frac{m}{2}}-1\right)-\frac{m}{2},\text{ if } m\equiv 2\hspace{1mm}(mod\hspace{1mm}4)\text{ and } m>6;  \\
                6m-1,\text{ if } m=6 .
		\end{cases} &
	\end{flalign*}
	Moreover, the code $\mathcal{C}_{s}$ has parameters $[2^{m}-1, 2^{m}-1-L_{s}, d(\mathcal{C}_{s})]$, where
    \begin{align*}
		 & \begin{cases}
			\max\{7, 2^{(m-2)/2}+1\}\leq d(\mathcal{C}_{s}) \leq 1+m\left(2^{\frac{m}{2}}+1\right)-\frac{m}{2} &\text{ if }  m\equiv 0 \hspace{1mm}(\textnormal{mod}\hspace{1mm}4)  \\
			2^{(m-2)/2}+1 \leq d(\mathcal{C}_{s}) \leq 1+m\left(2^{\frac{m}{2}}-1\right)-\frac{m}{2} &\text{ if } m\equiv 2 \hspace{1mm}(\textnormal{mod}\hspace{1mm}4)\text{ and } m>6
		\end{cases}
	\end{align*}
    \begin{proof}
        With the help of Lemma \ref{L4} and Eq. (\ref{E3}), proceeding similarly to Lemma \ref{L16}, we obtain
        \begin{equation}\label{EF9}
		\operatorname{Tr}\left(\sum_{i=1}^{2^{h}-1}x^{i}\right)=\operatorname{Tr}\left(\sum_{i\in\Gamma_{(h-1)}}x^{i+2^{h-1}}+\sum_{i\in\Gamma_{(h-3)}}x^{i+2^{h-3}}+\dots+\sum_{i\in\Gamma_{(3)}}x^{i+2^{3}}+x^{3}\right)
	\end{equation}
	if $h$ is even; and
	\begin{equation}\label{EF10}
		\operatorname{Tr}\left(\sum_{i=1}^{2^{h}-1}x^{i}\right)=\operatorname{Tr}\left(\sum_{i\in\Gamma_{(h-1)}}x^{i+2^{h-1}}+\sum_{i\in\Gamma_{(h-3)}}x^{i+2^{h-3}}+\dots+\sum_{i\in\Gamma_{(2)}}x^{i+2^{2}}+x\right)
	\end{equation}
	if $h$ is odd.
    \vskip 1pt
    Note that
    \begin{align}\label{EF11}
			\operatorname{Tr}\left(\sum_{i=1}^{2^{h}-1}x^{i+2^{h+1}}\right) &=
			\operatorname{Tr}\left(\sum_{i\in \Gamma_{(h)}}x^{i+2^{h+1}}\right)+\operatorname{Tr}\left(\sum_{i\in \bar{\bar{A}}\backslash\Gamma_{(h)}}^{}x^{i+2^{h+1}}\right) \nonumber \\
            &=
			\operatorname{Tr}\left(\sum_{i\in \Gamma_{(h)}}x^{i+2^{h+1}}\right)+\operatorname{Tr}\left(\sum_{i=1}^{2^{h-1}-1}x^{i+2^{h}}\right) \nonumber \\
			&= \operatorname{Tr}\left(\sum_{i\in \Gamma_{(h)}}x^{i+2^{h+1}}+\sum_{i\in \Gamma_{(h-1)}}x^{i+2^{h}}\right)+\operatorname{Tr}\left(\sum_{i=1}^{2^{h-2}-1}x^{i+2^{h-1}}\right) \nonumber \\
            &= \operatorname{Tr}\left(\sum_{i\in \Gamma_{(h)}}x^{i+2^{h+1}}+\sum_{i\in \Gamma_{(h-1)}}x^{i+2^{h}}+\sum_{i\in \Gamma_{(h-2)}}x^{i+2^{h-1}}\cdots+\sum_{i\in \Gamma_{(2)}}x^{i+2^{3}}+\sum_{i\in \Gamma_{(1)}}x^{i+2^{2}}\right)
		\end{align}
        When $h$ is even, with the help of Eq. (\ref{EF8}), (\ref{EF9}) and (\ref{EF11}), we have
        \begin{align}\label{EF12}
		\operatorname{Tr}\left(f_{8}(x+1)\right)&=\operatorname{Tr}\left(\sum_{i\in\Gamma_{(h)}}x^{i+2^{h+1}}+\sum_{i\in\Gamma_{(h-1)}}x^{i+2^{h}}+\sum_{i\in\Gamma_{(h-1)}\setminus\Gamma_{(h-2)}}x^{i+2^{h-1}}+\sum_{i\in\Gamma_{(h-3)}}x^{i+2^{h-2}}+\dots\right. \nonumber \\ &\mathrel{\phantom{=}} \left.\kern-\nulldelimiterspace +\; \sum_{i\in\Gamma_{(3)}\setminus\Gamma_{(2)}}x^{i+2^{3}}+\sum_{i\in\Gamma_{(1)}}x^{i+2^{2}}+x^{3}+x\right)
	\end{align}
    and, when $h$ is odd, with the help of Eq. (\ref{EF8}), (\ref{EF10}) and (\ref{EF11}), we have
    \begin{align}\label{EF13}
		\operatorname{Tr}\left(f_{8}(x+1)\right)&=\operatorname{Tr}\left(\sum_{i\in\Gamma_{(h)}}x^{i+2^{h+1}}+\sum_{i\in\Gamma_{(h-1)}}x^{i+2^{h}}+\sum_{i\in\Gamma_{(h-1)}\setminus\Gamma_{(h-2)}}x^{i+2^{h-1}}+\sum_{i\in\Gamma_{(h-3)}}x^{i+2^{h-2}}+\cdots \right. \nonumber \\ &\mathrel{\phantom{=}} \left.\kern-\nulldelimiterspace +\; \sum_{i\in\Gamma_{(2)}\setminus\Gamma_{(1)}}x^{i+2^{2}}\right)
	\end{align}
    From Lemma \ref{lem:L1}, \ref{EF15} and \ref{EF14}, it is evident that none of the terms on the right-hand side of Eq. (\ref{EF12}) and (\ref{EF13}) will mutually cancel out.
    \vskip 0pt
    Note that for any integer $t\geq 1$, $|\Gamma_{(t)}|=|\Gamma_{(t+1)}\backslash\Gamma_{(t)}|=2^{t-1}$ and $|C_{1+2^{h}}|=h=\frac{m}{2}$.
    If $h\geq 4$ is even, by Lemma \ref{lem:L1}, \ref{EF15} and Eq. (\ref{EF12}), we have the linear span of $s^{\infty}$ as follows
    \begin{flalign*}
			\hskip 20pt	L_{s}&= \left(2^{h-1}+2^{h-2}+2^{h-3}+\dots+2+1\right)\cdot m +2m -\frac{m}{2}& \\
			&=m(2^{h}+1)-\frac{m}{2}.
		\end{flalign*}
        If $h\geq 4$ is odd, by Lemma \ref{lem:L1}, \ref{EF15} and Eq. (\ref{EF13}), we have the linear span of $s^{\infty}$ as follows
        \begin{flalign*}
			\hskip 20pt	L_{s}&= \left(2^{h-1}+2^{h-2}+2^{h-3}+\dots+2+1\right)\cdot m -\frac{m}{2}& \\
			&=m(2^{h}-1)-\frac{m}{2}.
		\end{flalign*}
        For $h=3$, note that $|C_{5+2^{h+1}}|=h-1=2$ by Lemma \ref{EF15}, and $|C_{1+2^{h}}|=h$. Then the linear span of $s^{\infty}$ is as follows 
        \begin{flalign*}
			\hskip 20pt	L_{s}&= \left(2^{2}+2+1\right)\cdot m -\frac{m}{2}-1-\frac{m}{2}& \\
			&=6m-1.
		\end{flalign*}
        Therefore, from Lemma \ref{lem:technical} and Eq. (\ref{EF12}), (\ref{EF13}) we get the result on the generator polynomial corresponding to the sequence $s^{\infty}$.
        \vskip 0pt
        We now prove the result on the lower bound of the minimum weight $d(\mathcal{C}_{s})$. It is easy to check that the reciprocal of the generator polynomial $\textit{g}_{s}(x)$ has roots $\alpha^{j}$ for all $j$ in $\{1+2^{h+1},3+2^{h+1},\cdots,2^{h}-1+2^{h+1}\}$. Since, the code $\mathcal{C}_{s}$ generated by $\textit{g}_{s}(x)$ and the code generated by the reciprocal of $\textit{g}_{s}(x)$ have identical weight distribution, the minimum weight $d(\mathcal{C}_{s})\geq 2^{h-1}+1$ by the help of the Hartmann-Tzeng bound. When $h$ is even, the code $\mathcal{C}_{s}$ is the subcode of the cyclic code generated by $m_{\alpha^{-1}}(x)m_{\alpha^{-(2^{k}+1)}}(x)m_{\alpha^{-(2^{2k}+1)}}(x)$, where $k=h+1$. As $\operatorname{gcd}(k,m)=1$, $\mathcal{C}_{s}$ is a triple-error-correcting code $(\text{see Theorem 1 in \cite{SETA13}})$. The upper bound of the minimum weight $d(\mathcal{C}_{s})$ follows from the Singleton bound. By combining these facts, we get the desired conclusion.
    \end{proof}
\end{theorem}
\vspace{0.5em}
\begin{example}
    Let $m=6$ and $\alpha$ be the root of the primitive polynomial $x^6 + x + 1$ over $\mathbb{F}_{2}$. The generator polynomial of $\mathcal{C}_{s}$ is $\textit{g}_{s}(x)=x^{35} + x^{34} + x^{30} + x^{27} + x^{25} + x^{24} + x^{22} + x^{19} + x^{15} + x^{13} + x^9 + x^7 + x^6 +
    x^5 + x^4 + x^2 + 1$. Then $\mathcal{C}_{s}$ is a binary $[63,28,9]$ cyclic code and $\mathcal{C}_{s}^{\perp}$ is a $[63,35,10]$ cyclic code.
\end{example}
\vspace{0.5em}
\begin{example}
    Let $m=8$ and $\alpha$ be the root of the primitive polynomial $x^8 + x^4 + x^3 + x + 1$ over $\mathbb{F}_{2}$. The generator polynomial of $\mathcal{C}_{s}$ is $\textit{g}_{s}(x)=x^{132} + x^{131} + x^{127} + x^{126} + x^{125} + x^{122} + x^{116} + x^{115} + x^{114} + x^{112} + x^{111} + x^{110} + x^{104} + x^{103} + x^{102} + x^{97} + x^{96} + x^{94} + x^{89} + x^{88} + x^{86} + x^{80} + x^{75} + x^{74} + x^{72} + x^{68} + x^{67} + x^{66} + x^{61} + x^{56} + x^{55} + x^{52} +
    x^{50} + x^{49} + x^{45} + x^{43} + x^{42} + x^{38} + x^{30} + x^{28} + x^{24} + x^{20} + x^{18} + x^{16} + x^{15} + x^{12} + x^{10} + x^6 + x^4 + x + 1$. Then, by using a Magma program, we have $\mathcal{C}_{s}$ is a binary $[255,123,d(\mathcal{C}_{s})]$ cyclic code, where $20\leq d(\mathcal{C}_{s})\leq 31$ and $\mathcal{C}_{s}^{\perp}$ is a $[255,132,d(\mathcal{C}_{s}^{\perp})]$ cyclic code, where $22\leq d(\mathcal{C}_{s}^{\perp})\leq 24$.
\end{example}
\vspace{0.5em}
\begin{remark}
	\textcolor{black}{It can be seen that the trinomials $f_{6}(x)=x+x^{2^{m/2}}+x^{2^m-2^{m/2}+1}$, $f_{7}(x)=x+x^{2^{m/2+1}-1}+x^{2^{m}-2^{m/2+1}+2}$ in case of $m\equiv 2\hspace{1mm}(\textnormal{mod}\hspace{1mm}4)$ and $f_{8}(x)=x+x^{2^{m/2}}+x^{2^{m-1}-2^{m/2-1}+1}$ in case of $m\equiv 0\hspace{1mm}(\textnormal{mod}\hspace{1mm}6)$ are not permutations over $\mathbb{F}_{2^{m}}$. When $m=6$, the code $\mathcal{C}_{s}$ designed by $f_{6}(x)$, $f_{7}(x)$ and $f_{8}(x)$ have parameters $[63,39,7]$, $[63,45,5]$ and $[63,28,9]$, respectively, while the best known linear codes in the Database $\cite{SETA23}$ has parameters $[63,39,9]$, $[63,45,8]$ and $[63,28,15]$, respectively. It should be noted that although the trinomials in Table \ref{Table1} are permutations over $\mathbb{F}_{2^m}$ for certain values of $m$, for $m>6$ the codes $\mathcal{C}_{s}$ and $\mathcal{C}_{s}^{\perp}$ do not guarantee optimality. For $m\geq 8$, the parameter of $\mathcal{C}_{s}$ is extensive. Due to the vast computation required, verifying the minimum weight of $\mathcal{C}_{s}$ using a Magma program becomes difficult. Therefore, paying more attention to developing tighter lower and upper bounds on the minimum distance or selecting suitable trinomials with permutation property (or low-differential uniformity) that provide the minimum distance of the code $\mathcal{C}_{s}$ closer to the square-root bound would be beneficial.} 
\end{remark}
\section{Summary and concluding remarks}
\textcolor{black}{Fascinated by the work of Ding \cite{SETA5} and the joint work of Ding and Zhou \cite{SETA3}, we have investigated some known families of permutation trinomials over $\mathbb{F}_{2^m}$ and constructed several infinite families of binary cyclic codes of length $2^m-1$ with dimensions larger than $(2^m-1)/2$ and minimum distance closer to the square-root bound. Some of the families of codes are distance-optimal. We determined the upper bound of the minimum distances of these codes. The main results of this paper demonstrate that suitable permutation monomials and trinomials can be used for the construction of cyclic codes with desirable parameters. Readers interested in working on this topic are invited to develop tighter upper and lower bounds of the minimum distances or to find new strategies in determining the linear span of sequences in constructing codes with minimum distance closer to the square-root bound by employing suitable polynomials.  
\vskip 1pt
Some families of binary cyclic codes presented in this paper are closely related to the triple-error-correcting binary primitive BCH codes; they could be used in constructing quantum codes \cite{SETA19, SETA18}.}



\begin{thebibliography}{1}
	
	
	\bibitem{SETA1}
	Antweiler~M., Bomer~L.: 
	\newblock Complex sequences over $GF(p^M)$ with a two-level autocorrelation function and a large linear span.
	\newblock { IEEE Trans. Inf. Theory} \textbf{38}(1), 120–130 (1992).
	
	\bibitem{SETA2}
	Ding~C., Qu~L., Wang~Q., Yuan~J., Yuan~P.:
	\newblock Permutation trinomials over finite fields with even characteristic.
	\newblock {SIAM J. Discret. Math.} \textbf{29}(1), 79–92 (2015).
	
	\bibitem{SETA3}
	Ding~C., Zhou~Z.: 
	\newblock Binary cyclic codes from explicit polynomials over $GF(2^m)$.
	\newblock {Discret. Math.}  \textbf{321}, 76–89 (2014).
	
	\bibitem{SETA4}
	Dobbertin~H.: 
	\newblock Almost perfect nonlinear power functions on $GF(2^n)$: the Welch case. \newblock {IEEE Trans. Inf.
 Theory} \textbf{45}(4), 1271–1275 (1999).
	
	\bibitem{SETA5}
	Ding~C.: 
	\newblock Cyclic Codes from some monomials and trinomials.
	\newblock {SIAM J. Discret. Math.} \textbf{27}(4), 1977–1994 (2013).
	
	\bibitem{SETA6}
	Kasami~T.: 
	\newblock  The weight enumerators for several classes of subcodes of the 2nd order binary Reed-Muller codes.
	\newblock {Inf. Control} \textbf{18}, 369–394 (1971).
	
	\bibitem{SETA7}
	Si~W., Ding~C.:
	\newblock A simple stream cipher with proven properties.
	\newblock { Cryptogr. Commun.} \textbf{4}, 79-104 (2012).
	
	\bibitem{SETA8}
	Mesnager~S., Shi~M., Zhu~H.:
	\newblock Study of cyclic codes from low differentially uniform functions and its consequences.
	\newblock {Discret. Math.} \textbf{347}, 114033 (2024).
	
	\bibitem{SETA9}
	Li~L., Zhu~S., Liu~L., Kai~X.:
	\newblock Some $q$-ary cyclic codes from explicit monomials over $\mathbb{F}_{q^{m}}$.
	\newblock {Probl. Inf. Transm.} \textbf{55}(3), 254-274 (2019).
	
	\bibitem{SETA10}
	Rajabi~Z., Khashyarmanesh~K.:
	\newblock Some cyclic codes from some monomials. \newblock{Appl. Algebra Engrg. Comm. Comput.} \textbf{28}, 469-495 (2017).
	
	\bibitem{SETA11}
	Tang~C., Qi~Y., Xu~M.:
	\newblock A note on cyclic codes from APN functions.
	\newblock {Appl. Algebra Engrg. Comm. Comput.} \textbf{25}, 21-37 (2014).


    \bibitem{SETA13}
	Bracken~C., Helleseth~T.:\newblock Triple-Error-Correcting BCH-Like Codes. In:\newblock {IEEE International Symposium on Information Theory}, pp. 1723-1725. IEEE Xplore, Seoul, Korea (South) (2009).
    
	
		\bibitem{SETA12}
		Ding~C.:
		\newblock A sequence construction of cyclic codes over finite
 fields. \newblock{arXiv:1611.06487v2} (2024). Available: \url{http://arxiv.org/abs/1611.06487}.
	
	
	
		\bibitem{SETA14}
		Ding~C.:
		\newblock A sequence construction of cyclic codes over finite fields.	\newblock {Cryptogr. Commun.} \textbf{10}, 319-341 (2018).

        
	\bibitem{SETA15}
	Li~K., Qu~L., Chen~X.:
	\newblock New classes of permutation binomials and permutation trinomials over finite fields.
	\newblock {Finite Fields and Appl.} \textbf{43}, 69-85 (2017).
    
	
	\bibitem{SETA16}
	Cherowitzo~W.:
	\newblock $\alpha$-Flocks and hyperovals.
	\newblock {Geom. Dedicata} \textbf{72}, 221-246 (1998).
	
	\bibitem{SETA17}
	Dobbertin~H.:
	\newblock Uniformly representable permutation polynomials. In:\newblock {Proceedings of SETA 01, T. Helleseth, P.V. Kumar, and K. Yang, eds., Springer, London}, pp. 1-22, (2002).
	
	\bibitem{SETA18}
	Shi~X., Yue~Q., Wu~Y.:
	\newblock The dual-containing primitive BCH codes with the maximum designed distance and their applications to quantum codes.
	\newblock {Des. Codes Cryptogr.} \textbf{87}, 2165–2183 (2019).
	
	\bibitem{SETA19}
	Aly~S.~A., Klappenecker~A.:
	\newblock On Quantum and Classical BCH Codes.
	\newblock {IEEE Trans. Inf. Theory} \textbf{53}, 1183–1188 (2007).
	
	\bibitem{SETA20}
	Hartmann~C.~R., Tzeng~K.~K.:
	\newblock Generalizations of the BCH Bound.
	\newblock {Inf. Control} \textbf{20}, 489–498 (1972).

    	\bibitem{SETA24}
		Helleseth~T., Li~C., Xia~Y.:
		\newblock Investigation of the permutation and linear codes from the Welch APN function.
		\newblock { Des. Codes Cryptogr.} 1-23 (2024).
	
		\bibitem{SETA21}
		Singleton~R.:
		\newblock Maximum distance $q$-nary codes.
		\newblock {IEEE Trans. Inf. Theory} \textbf{10}, 116–118 (1964).

        \bibitem{SETA22}
		Xie~X., Zhao~Y., Sun~Z., Zhou~X.:
		\newblock Binary $[n,(n\pm 1)/2]$ cyclic codes with good minimum distances from sequences.
		\newblock {Discret. Math.} \textbf{348}, 114369 (2025).
	
	
	\bibitem{SETA23}
	Grassl~M.: Bounds on the minimum distance of linear codes and quantum codes, Online available at \url{http://www.codetables.de/}. 
	
\end{thebibliography}

\end{document}